\documentclass{vldb}
\usepackage{pgfplots}
\pgfplotsset{compat=newest}
\usepackage{balance}  % for  \balance command ON LAST PAGE  (only there!)
\usepackage{amsmath, amscd}
\usepackage{placeins}
\usepackage{mathtools} 
\usepackage{fancybox}
\usepackage[linesnumbered]{algorithm2e}
\usepackage{graphicx}
\usepackage{comment}
\usepackage{setspace}
\usepackage{fancyhdr} 
\usepackage{graphicx} 
\usepackage[export]{adjustbox}
\usepackage{framed}
\usepackage{tabularx} 
\usepackage{multicol} 
\usepackage{stmaryrd}
\usepackage{amsfonts} 
\usepackage{enumitem}
\usepackage{pifont} 
\usepackage{url} 
\usepackage{times} 
\usepackage{verbatim}
\usepackage{booktabs} 
\usepackage{longtable} 
\usepackage{pbox}
\usepackage{subfigure}
\usepackage{color}
\usepackage{etex}

\usepackage{listings}
\usepackage{tikz}
\usetikzlibrary{shapes,shadows}
\usetikzlibrary{arrows,patterns,
shadows,positioning,fit,shapes,calc,decorations.pathmorphing}
\usetikzlibrary{positioning,chains,fit,shapes,calc}
\usetikzlibrary{arrows,automata}
\usepackage{tikz-cd}
\usepackage{xspace}
\usepackage[font={small,it}]{caption}
\usepackage{pgfplots}
\pgfplotsset{compat=newest}
\usepgfplotslibrary{fillbetween}

\definecolor{light-gray}{gray}{0.90}
\definecolor{MyBlue}{rgb}{0.00,0.25,0.5}
\colorlet{Shade}{light-gray}

\newcommand\souffle{\textsc{Souffl\'e}\xspace}

\newtheorem{theorem}{Theorem}
\newtheorem{definition}{Definition}
\newtheorem{lemma}{Lemma}

\newtheorem{example}{Example}

\newcommand\com{,}

\newcommand\bigo{\mathcal O}

\definecolor{myblue}{RGB}{80,80,160}
\definecolor{mygreen}{RGB}{80,160,80}
\definecolor{myred}{RGB}{160,80,80}

\newcommand\rdom{\mathbb D}
\newcommand\nr{n}
\newcommand\tl{m}
\newcommand\aset{A_R}
\newcommand\rel{R}
\newcommand\euler{\mathsf{e}}

\newcommand\perm{\textsf{Pm}}
\newcommand\rangequery{\rho}
\newcommand\searchquery{\sigma}
\newcommand\lex{\ell} % a single lex order
\newcommand\lexset{L} % single lex set 
\newcommand\search{\mathit{s}} % single search e.g {x, y}

\newcommand\searchset{S} % set of searches 
\begin{document}

\sloppy

% ****************** TITLE ****************************************
\title{Optimal On The Fly Index Selection in Polynomial Time}
\subtitle{(Extended Report)}
\numberofauthors{3}
\author{\alignauthor Herbert Jordan \\ \affaddr{University of Innsbruck, Austria}
        \alignauthor Bernhard Scholz \\ \affaddr{University of Sydney, Australia} 
        \alignauthor Pavle Suboti\'{c} \\ \affaddr{University College London, United Kingdom}}

\maketitle

\begin{abstract}
The index selection problem (ISP) is an important problem for accelerating the execution 
of relational queries, and it has received a lot of attention as a combinatorial knapsack 
problem in the past. Various solutions to this very hard problem have been provided. 
In contrast to existing literature, we change the underlying assumptions of the
problem definition:  we  adapt the problem for systems that store relations 
in memory, and use complex specification languages, e.g., Datalog. In our framework, 
we decompose complex queries into \emph{primitive searches} that select tuples in a 
relation for which an equality predicate holds. A primitive search can be accelerated 
by an index exhibiting a worst-case run-time complexity of log-linear time in the size of the 
output result of the primitive search. However,
the overheads associated with maintaining indexes are very costly in terms of memory and 
computing time. 

In this work, we present an optimal polynomial-time algorithm that finds the minimal 
set of indexes of a relation for a given set of primitive searches. An index may cover more 
than one primitive search due to the algebraic properties of the search predicate, 
which is a  conjunction of equalities over the attributes of a relation. The index search 
space exhibits a complexity of $\bigo(2^{m^m})$ where  $m$ is the number of attributes
 in a relation, and, hence  brute-force algorithms searching for solutions in the index domain  are infeasible. 
 As a scaffolding for designing a polynomial-time algorithm,  we build a partial order on search operations and use a constructive version 
of Dilworth's theorem. We show a strong relationship between chains of primitive searches 
(forming a partial order) and indexes. We demonstrate the effectiveness and efficiency of our algorithm for an in-memory 
Datalog compiler that is able to process relations with billions of entries in memory.
\end{abstract}

\section{Introduction}
There has been a resurgence in the use of Datalog in several computer science communities including program 
analysis where it is used as a domain specific language (DSL) for succinctly specifying various classes of 
static analyses. In this setup, an input program is converted into an \emph{extensional database} (EDB) and 
the static analysis specification is encoded as an \emph{intensional database} (IDB). Such use cases, unlike 
traditional database queries, typically consist of hundreds of relations and hundreds of deeply nested 
rules~\cite{Doop}, and result in giga-tuple sized relations~\cite{CAV}. Consequently, several high performance 
Datalog engines~\cite{CAV,Hoder11,logicblox,Socialite,Bierhoff05} have been employed for 
performing such computations. These engines use Datalog purely as a computational vehicle 
and use bottom-up evaluation techniques that usually involve some degree of 
compilation~\cite{Socialite, CAV, Coral, GlueNail}. Moreover, to further improve performance, relations are stored as 
indexed-organized tables in-memory, hence, enabling improved cache behaviour, lookup complexity, etc. As a result of 
these design characteristics, such engines require atypical index selection techniques that result in improved 
run-time and memory consumption while not resulting in noticeable compilation overhead.

Traditionally relational database management systems have solved
the index selection problem (ISP)~\cite{Schkolnick1975,Comer1978,Ip, Genetic}
by variants of the $0$-$1$ knapsack problem. Such formulations implicitly solve two subproblems of 
query optimization, literal scheduling along side index selection. While these techniques such as~\cite{Chau97} 
are well established in the relational databases literature, they are too computationally expensive for 
large Datalog programs and are rarely used in most Datalog engines. Too simplify this task, modern Datalog engines 
often require users to provide annotations to guide the engine in the choice of indices~\cite{Socialite, Doop}. Again, for large 
Datalog programs such approaches are very cumbersome as they put the entire optimization burden on the
user, often resulting in painstaking trial and error process that results in far from optimal index performance. In this 
paper, we present a practical yet powerful middle ground approach: we solve the index selection optimally,
while leaving join scheduling to minimal user annotations (usually on a small set of rules) or by automated means~\cite{Selingers}. In 
particular, our method is ideal for compilation based engines as the index selection is performed \emph{on the fly}, and results in 
negligible compilation time overheads while considerably boosting performance.

Our index selection approach takes 
into account several factors  
present in high-performance Datalog engines: First, in these engines, relations are not normalized and a very 
large number of indices may occur making the problem combinatorially intractable. Second, an index represents a 
\emph{whole} relation and hence there is no need top capture maintenance costs in the ISP 
formulation as it can be assumed that all indexed relations have a uniform cost. Third, high-performance Datalog 
engines translate Datalog queries into intermediate representations~\cite{Socialite, CAV, Coral, GlueNail} that 
assume a fixed join order~\cite{Socialite, CAV} and are decomposed into a simpler relational algebra operators that 
operate on a single relation that we refer to \emph{primitive searches}. Therefore, ISP is computed for a single 
relation only.

As a consequence of the new assumptions, the index selection problem is reduced to the problem of 
minimizing the number of indices for each relation separately. One option is to formulate the problem as a 
\emph{Minimum Set Cover Problem} (MSCP), however, such formulations do not give us tractability, and MSCP is a too 
coarse as a combinatorial vehicle. On the surface it appears that finding a tractable solution is futile, considering 
the vast search space of all possible indices of a single relation (i.e., space of ordered attribute subsets). However, 
as indices need to be computed on the fly during compilation/interpretation time and by our assumption that each 
primitive search requires at least one index, the ISP problem can be reformulated into a covering problem, i.e.,  each 
primitive search needs an index cover. This new formulation, while more fitting to the problem at hand, also reveals deeper 
mathematical structures that 
allows us to solve the problem in polynomial-time. As a result, we are able to perform index minimization with 
negligible overhead and obtain significant speedups with minimal 
consumption.

The key to our index minimization approach is the formulation of a relationship between the space of indices and 
\emph{search chains}. In the space of search chains we are able to leverage existing combinatorial results i.e., 
Dilworth's theorem~\cite{Dilworth50} that provide polynomial-time 
algorithms to find an optimal solutions. Using our established relationship, we are able to 
convert the optimal search chain solution to an optimal set of indices that we use to construct 
indexed joins, refered to as \emph{range nested loop joins}. We further clarify our method by the motivating 
example below:

\begin{example}[Motivating]
\label{ex:motivating}
Assume we have a 
non-recursive Datalog rule that has 
a single ternary input relation $A(x,y,z)$ and a ternary output relation $B(x,y,z)$, where x, y and z are 
attributes. 
\begin{align*}
B(r,p,q) \leftarrow & A(r,p,q), A(q,\_,\_), \\ & A(p,q,\_), A(p,\_,q), A(q,p,r).
\end{align*}

This query is transformed to a \emph{nested loop join} version of the query as depicted below. The details of this transformation are 
explained in Section~\ref{sec:integrate}.

\begin{tabbing}
\hspace{0.9cm}\=\hspace{0.1cm}\=\hspace{0.1cm}\=\hspace{0.1cm}\=\hspace{0.1cm}\=\hspace{0.1cm}\=\hspace{0.1cm}\=\hspace{0.1cm}\=\hspace{0.1cm}\=\hspace{0.1cm}\=\kill 
$\textit{loop}_1$: \> $\textbf{for all} \, t_1 \in A  \, \textbf{do}$ \\
$\textit{loop}_2$: \>\> $\textbf{for all} \, t_2 \in \searchquery_{x = t_1(z)}(A) \, \textbf{do}$ \\
$\textit{loop}_3$: \>\>\> $\textbf{for all} \, t_3 \in \searchquery_{x = t_1(y), y = t_1(z)}(A)  \, \textbf{do}$ \\
$\textit{loop}_4$: \>\>\>\> $\textbf{for all} \, t_4 \in \searchquery_{x = t_1(y), z = t_1(z)}(A) \, \textbf{do}$ \\
$\textit{loop}_5$: \>\>\>\>\> $\textbf{for all} \, t_5 \in \searchquery_{x = t_1(z), y = t_1(y), z = t_1(x)}(A) \, \textbf{do}$ \\
\>\>\>\>\>\>\textbf{if} $ t_1 \not \in B$ \textbf{ then} \\
\>\>\>\>\>\>\>\textbf{add } $t_1$ \textbf{ to } $B$\\
\>\>\>\>\>\>\textbf{endif}\\
\>\>\>\>\>\textbf{endfor}\\
\>\>\>$\ldots$\\
\>\textbf{endfor}
\end{tabbing}

From the nested loop join we extract primitive searches denoted by $\searchquery_{\varphi}$ where $\varphi$ is an equality 
predicate between attributes and contants bound by tuple values obtained further above the loop. To improve performance we use 
indices which we abstract as $\ell$. A na\"ive approach assigns each query an index which we represent by a 
lexicographical order, i.e., $\sqsubseteq_\ell$ where $\ell$ is a sequence of attributes $ \ell = x \prec \dots \prec z$ that 
implements a lexicographical order a binary search tree. $c_x$ denotes a constant for attribute $x$ obtained from 
a tuple $t_1$, i.e., $t_1(x)$. We assign four indices to the primitive searches as shown the table below:
 
\begin{center}
\footnotesize
  \begin{tabular}{l|l|l}
  & \multicolumn{2}{c}{} \\
Query component & Primitive Search &   $\ell$ \\
  \hline 
$A(q, \_, \_)$ & $\searchquery_{x = t_1(z)}(A)$                      &  $x$ \\
$A(p,q,\_)$    & $\searchquery_{x = t_1(y), t_1(z)}(A)$             &  $x \prec y$\\
$A(p,\_,q)$    & $\searchquery_{x = t_1(y), z = t_1(z)}(A)$             &  $x \prec z$\\
$A(q,p,r)$     & $\searchquery_{x = t_1(z), y = t_1(y), z = t_1(x)}(A)$    &  $x \prec y \prec z$ \\ 
 \end{tabular}
\end{center}

While this assignment of indices will speed up the join computation, it is not the most optimal assignment.
The research question is therefore, how can we assign indices to primitive searches in the most optimal 
way.

\end{example}

To demonstrate the practicality of our technique, we have implemented the techniques discussed in this 
paper in an open source, high-performance Datalog engine called \souffle~\cite{CAV} 
that is used for large-scale static program 
analyses hundreds of rules and relations and processing up to giga-tuples of data.

We have performed experiments on a wide variety of program analysis specifications~\cite{Doop} and a diverse 
set of datasets that include the OpenJDK\footnote{Available from \url{http://www.oracle.com/technetwork/oracle-labs/datasets/overview/index.html}}, a large industrial benchmark from Oracle. Our technique results in considerable improvements compared to several 
alternative indexing schemes. Our experiments suggest that our approach gives considerable run-time and memory usage improvements
compared to \souffle's other indexing schemes. Moreover, with our technique \souffle is able to analyze problems typically deemed 
too difficult for Datalog-based tools on a par with the state-of-the-art hand crafted analyzer presented in~\cite{OOPSLA15}.

Our contributions are summarized as follows: 
\begin{itemize}
\item We describe a join mechanism that allows for optimal index selection
\item We formally define the minimal index selection problem (MOSP), including its state space.
\item We introduce a novel polynomial time algorithm to solve MOSP via search chains
\item We present a case study implementing MOSP in \souffle, an open-source Datalog engine for large 
scale program analysis. We demonstrate the effectiveness of MOSP in \souffle with large input 
instances and several real world analyses.
\end{itemize}

The paper is organized as follows:
We begin with preliminary definitions in Section~\ref{sec:prelims}. In Section~\ref{sec:cover} we give an overview of the 
join computation we perform. Section~\ref{sec:mosp} formally 
states the MOSP problem and analyses the search space of MOSP, finally leading to an optimal algorithm for solving MOSP.
In Section~\ref{sec:integrate} we discuss how our algorithm can potentially be integrated into databases and other query engines.
In Section~\ref{sec:experiments} we introduce a case study, where we implement our approach in a Datalog-based program analysis engine and 
evaluate our algorithm improvements observed using our technique. We highlight related work in Section~\ref{sec:relatedwork} and draw 
relevant conclusions in Section~\ref{sec:conclusion}.

\section{Preliminaries} 
\label{sec:prelims}
%
%\paragraph{Sets and Permutations} 
%We write $a \in X$, if element $x$ is in set $X$. We write $a \not \in
%X$, if element $x$ is not in set $X$. The \emph{cardinality} of a set
%$X$ counts the elements in set $X$ and is denoted by $|X|$.  Relation
%$Y \subseteq X$ holds, iff $Y$ is a subset of $X$.  
A \emph{power set}
of set $X$ is the set of all subsets of $X$ and is denoted by $2^X
$. % The cardinality of power set $2^X$ is $2^{|X|}$.  
The \emph{cartesian product} of two sets $X$ and $Y$ is the set of all
pairs $(a,b)$ such that $a \in X$, and $b \in Y$, and is denoted by $X
\times Y$. The cardinality of the Cartesian product is $|X| \cdot
|Y|$. The \emph{finite n-ary cartesian} product is written as $X_1
\times \ldots \times X_k = \{(x_1, \ldots, x_k) : x_i \in X_i \}$
where elements are referred to as \emph{tuples} and are defined as nested
ordered pairs, i.e., $(X_1 \times \ldots \times X_{n-1} ) \times X_n$.
The \emph{permutations} of set $X=\{a_1, \ldots, a_k\}$ is the set of
all possible sequences formed by elements of set $X$ such that each
element occurs exactly once. The cardinality of $\perm(X)$ is given by the factorial of the cardinality
of set $X$, i.e., $|\perm(X)| = |X|!$. We define a sequence as $a_1 \prec a_2 \prec  \ldots \prec a_k$ where $\prec$ 
denotes a chaining of elements to form a sequence. 
%i.e., $$\perm(X) =   \{ a_1 \prec a_2 \prec  \ldots \prec a_k, a_2 \prec a_1 \prec  \ldots \prec a_k, \ldots \}$$

%\paragraph{Relations and Named Relations.} 
A \emph{relation} $\rel \subseteq \rdom_1
\times \ldots \times \rdom_{\tl}$ is a set of tuples $\{ t_1,
\ldots, t_\nr\}$ where $\nr$ is the number of tuples in the relation,
$\tl$ is the \emph{tuple length}, and $\rdom_i$ are the \emph{domains}
of the relation. A tuple $t$ is a fixed-length vector $\langle e_1,
e_2, ..., e_\tl \rangle$ whose elements $e_i$ are elements of domain
$\rdom_i$, for all $i$, $1\leq i \leq \tl$. 

A \emph{named relation} is a relation $\rel$ that uses \emph{attributes} 
to refer to specific element positions. The set of attributes  
$\aset=\{x_1, \ldots, x_\tl\}$ are $\tl$ distinct symbols and we
write $\rel (x_1, \ldots,$ $ x_\tl)$ to associate symbol $x_i$ to the
$i$-th position in the tuple. The elements of tuple $t=\langle e_1,
\ldots, e_\tl \rangle$ can be accessed by \emph{access function}
$t(x_i)$, which maps tuple $t$ to element $e_i$. E.g., given relation
$\rel(x, y, z)$ and a tuple $t =\langle e_1, e_2, e_3 \rangle \in
\rel$, the access function is $\{t(x) \mapsto e_1$, $t(y) \mapsto e_2$
and $t(z) \mapsto e_3 \}$. 

%We have an \emph{attribute function}
%$\iota: {\mathbb N} \rightarrow \aset$, that uniquely maps
%element positions of a tuple to attributes in a range between $1$ and
%$m$. E.g., for relation $\rel(x,y,z)$, the attribute function is
%$\iota= \{ 1 \mapsto x, 2 \mapsto y, 3 \mapsto y \}$. We reserve the symbols 
%$x, y$ and $z$ for attributes only. 

%\paragraph{Binary Relations and Orders.}
\newcommand\crel{\not \rel} A \emph{binary relation} $\rel \subseteq
\rdom \times \rdom$ is a set of ordered pairs. Two elements $a$ and
$b$ are related in $\rel$ written as $a \rel b$, if there is a pair
$(a,b) \in \rel$; two elements $a$ and $b$ are unrelated written as $a
\crel b$, if $(a,b) \not \in \rel$.  A relation is \emph{reflexive},
if $a \rel a$ for all elements $a$; \emph{symmetric}, if $a \rel b$
implies $b \rel a$ for all $(a,b) \in \rdom \times \rdom$;
\emph{asymmetric} if $a \rel b$ and $b \rel a$ implies $ a = b$,
\emph{transitive} if $a \rel b$ and $ b \rel c$ implies $a \rel c$, and
\emph{total} if $a \rel b$ or $b \rel a$ for all $(a,b) \in \rdom
\times \rdom$.

A binary relation $\leq$ is a \emph{pre-order} if the relation is
reflexive and transitive, a \emph{partial order} if the relation is
reflexive, asymmetric, and transitive, and a \emph{total order}\footnote{Sometimes a total order is also referred
  to as linear order, simple order, or (non-strict) ordering.} if the
relation is a partial order and is total. 

A lexicographical order $\sqsubseteq_\ell  {\mathcal D} \times {\mathcal D}$  is a total order defined 
over the domain of a relation where ${\mathcal D} = \rdom_1 \times \ldots \times \rdom_m$ 
is a finite n-ary cartesian product of the element domains, and the sequence $\ell \in \lexset$ 
is formed by a subset of attributes where each attribute occurs at most once
in the sequence.

\section{Computing Indexed Joins}
\label{sec:cover}
A major performance consideration in a Datalog engine is how join 
computations are performed. A join in traditional databases is computed by converting a Datalog 
rule to a \emph{primitive nested loop join}. The na\"ive assumption is that there is no underlying 
tuple order inside the nested loop join resulting in linear search time complexity.  

Primitive nested loop joins are defined in Fig.~\ref{fig:prim-nested loop join}. We refer to the head 
atom of a Datalog clause as $R_{0}$,
and each body atom as $R_{i}$ where $i > 0$. We partition the sequence of body atoms at a position 
index $k$ into positive and negative occurrences (i.e., negative if it is negated in the body), and 
denote positive atoms as $R_i^{+}$ where $0 < i \leq k$ and negative as $R_i^{-}$ where $i > k$, 
where $i$ denotes a position in the body.

In the primitive nested loop join we iterate (denoted by the \textbf{for all} construct) 
over tuples. The tuples are obtained from a filter called a \emph{primitive search} defined in 
Def.~\ref{def:ps} for positive relations. This comes from the implicit universal quantification in a 
Datalog clause. A primitive search extracts all tuples from a relation that adhere to a 
\emph{primitive search predicate}, i.e., a predicate limited to equalities of left-hand-side attributes and 
right-hand-side constants bound to tuples further up the nested loop join. Negative occurring atoms are tested for 
emptiness w.r.t. a primitive search on already stable relations. This semantics stems from the implicit 
non-existence quantification on attributes of negative body literals in a Datalog rule. The most inner 
operation in a nested loop join projects ($\pi$) the selected tuple into the head atom if the tuple does not already 
exist in the relation. This existence check is performed to ensure that tuples are not inserted twice into 
a relation, i.e., it enforces the set constraints for relational tables. At the primitive program level, several 
optimizations can now be performed. For example, the join can be marked for parallelisation directives, loops which 
have primitive searches subsumed by another loop can be coalesced into a single loop, loops can be 
pushed to the most outer possible layer (know as hoisting/layering)~\cite{CC}.

\begin{figure}[ht]
\begin{tabbing}
\hspace{1.5cm}\=\hspace{0.3cm}\=\hspace{0.3cm}\=\hspace{0.3cm}\=\hspace{0.3cm}\=\hspace{0.3cm}\=\hspace{0.3cm}\=\hspace{0.3cm}\=\kill 
$\textit{loop}_1$: \> $\textbf{for all} \, t_1 \in \rel^{+}_{1} \, \textbf{do}$ \\
$\textit{loop}_2$: \>\> $\textbf{for all} \, t_2 \in \sigma_{\varphi_2(t_1, t_2)}(\rel^{+}_{2})$ : \textbf{do} \\
\>\>\> $\ldots$\\
$\textit{loop}_k$: \>\>\>\> \textbf{for all} $t_k \in \sigma_{\varphi_k(t_1, t_2, \ldots, t_k)}(\rel^{+}_{k})$ :  \textbf{ do}\\ 
\>\>\>\>\>\textbf{if} $\sigma_{\varphi_{k+1}(t_1, t_2, \ldots, t_k)}(\rel^{-}_{k+1}) = \emptyset$ \textbf{ then} \\
\>\>\>\>\> $\ldots$\\
\>\>\>\>\>\>\textbf{if} $\sigma_{\varphi_{n}(t_1, t_2, \ldots, t_k)}(\rel^{-}_{n}) = \emptyset$ \textbf{ then} \\
\>\>\>\>\>\>\>\textbf{if} $\pi(t_1, \ldots, t_k) \not \in R_{0}$ \textbf{ then} \\
\>\>\>\>\>\>\>\>\textbf{add } $\pi(t_1, \ldots, t_k)$ \textbf{ to } $R_{0}$\\
\>\>\>\>\>\>\>\textbf{endif}\\
\>\>\>\>\> $\ldots$\\
\>\>\>\>\>\>\textbf{endif}\\
\>\>\>\>\textbf{endfor}\\
\>\>\> $\ldots$\\
\>\>\textbf{endfor}\\
\>\textbf{endfor}
\end{tabbing}
\caption{Nested-Loop Join: a relational algebra query is translated by a query planner to nested loop join; each 
loop enumerates tuples of a relation and filters selected tuples $t_1, \ldots, t_k$}
\label{fig:prim-nested loop join}
\end{figure}

\begin{definition}[Primitive Search]
\label{def:ps}
A primitive search has the following form: 
$$\searchquery_{x_1=v_1, \ldots, x_k=v_k} (R_{i}) \equiv \{ t \in R_{i} \mid t(x_1)=v_1 \wedge \ldots t(x_k) = v_k \}$$ where $R_{i}$ is a relation and $x_1=v_1, \ldots, x_k=v_k$ is the search predicate of the
relation where  $x_1, \ldots, x_k$ are variables (also known as attributes) of the
relation and $v_1, \ldots, v_k$ are either constants or values obtained from other 
tuple elements in relations $R_{j}$, $0 < j < i$. As an alternative notation, 
we denote $\searchquery_{{\varphi}(t_1, \dots, t_k)}$ where $\varphi \equiv x_1=v_1, \ldots, x_k=v_k$ as 
the substitution of $t_1$ to $t_k$ for appropriate constants $v_1$ to $v_k$.
\end{definition}

To improve the join computation performance we emply indices to each primitive search. Our technique rests on the 
assumption that all primitive searches benifit from being indexed. We refer to this assumption as the 
\emph{Minimal Index Assumption} (MIA). the benifit of indices is that they introduce orders 
on tuples in relations so that tuple lookups can be performed efficiently using some notion of a balanced search 
tree, in which tuples can be found in logarithmic-time rather than in linear-time. To create an order
among tuples in a relation, tuples must be made comparable. Since a tuple may have several elements, 
an order is imposed by element-wise comparison using a permutation over a subset of attributes, i.e., if the
first elements produces a tie, the second elements are used and so forth. This comparison is also known as 
a \emph{lexicographical order}. We abstract away the underlying implementation 
details of an index with a attribute sequence $\ell$. 

\begin{figure}[ht]
\begin{tabbing}
\hspace{1.2cm}\=\hspace{0.3cm}\=\hspace{0.3cm}\=\hspace{0.3cm}\=\hspace{0.3cm}\=\hspace{0.3cm}\=\hspace{0.3cm}\=\hspace{0.3cm}\=\kill 
$\textit{loop}_1$: \> $\textbf{for all} \, t_1 \in \rel_{1}^{+} \, \textbf{do}$ \\
$\textit{loop}_2$: \>\> $\textbf{for all} \, t_2 \in \sigma_{\rho_2(\ell_2, a_2, b_2)}(\rel_{2}^{+}) \, \textbf{do}$ \\
\>\>\> $\ldots$\\
$\textit{loop}_k$: \>\>\>\> \textbf{for all} $t_k \in \sigma_{\rho_k(\ell_k, a_k, b_k)}(\rel_{k}^{+})$  \textbf{ do}\\ 
\>\>\>\>\>\textbf{if} $\sigma_{\rho_{k+1}(\ell_{k+1}, a_{k+1}, b_{k+1})}(\rel_{k+1}^{-}) = \emptyset$ \textbf{ then} \\
\>\>\>\>\> $\ldots$\\
\>\>\>\>\>\>\textbf{if} $\sigma_{\rho_n(\ell_n, a_n, b_n)}(\rel_{n}^{-}) = \emptyset$ \textbf{ then} \\
\>\>\>\>\>\>\>\textbf{if} $\pi(t_1, \ldots, t_k) \not \in R_{0}$ \textbf{ then} \\
\>\>\>\>\>\>\>\>\textbf{add } $\pi(t_1, \ldots, t_k)$ \textbf{ to } $R_{0}$\\
\>\>\>\>\>\>\>\textbf{endif}\\
\>\>\>\>\>\>\textbf{endif}\\
\>\>\>\>\>$\ldots$\\
\>\>\>\>\>\textbf{endif}\\
\>\>\>\>\textbf{endfor}\\
\>\>\> $\ldots$\\
\>\>\textbf{endfor}\\
\>\textbf{endfor}
\end{tabbing}
\caption{Loop-Nest: a relational algebra query is translated by a query planner to nested loop join; each 
loop enumerates tuples of a relation and filters selected tuples $t_1, \ldots, t_k$  using a range 
tuple predicate $\rho(\ell, a, b)$. The inner most loop projects ($\pi$) selected tuples 
$t_1, \ldots, t_k$ to a tuple for the output relation $R$. If the tuple does not exist, it will be 
added to the output relation.} 
\label{fig:range-nested loop join}
\end{figure}

An indexed nested loop join is refered to as a \emph{range nested loop join}. Range nested loop joins are 
similar to primitive nested loop joins only that they are further specialized 
to operate on \emph{range searches}. Range searches assume and index and 
hence assume that tuples are ordered, hence an 
ordered set of tuples exhibits a worst-case complexity for executing a range search 
in a linear-log time in the size of the output, i.e, $\bigo(|\sigma_{\rangequery(\ell,a,b)}| \log n)$ 
where $n$ is the number of tuples in a relation $\rel$. We define a range search in Def.~\ref{def:rs}.

\begin{definition}[Range Search]
\label{def:rs}
  A range  $\sigma_{\rangequery(\ell,a,b)}$ is defined for a relation
  $\rel \subseteq {\mathcal D}$ and its
  semantics is given by, \[ \sigma_{\rangequery(\ell,a,b)} (R) = \{ t \in R
  \mid a \, \sqsubseteq_\lex \, t \wedge t \, \sqsubseteq_\lex \, b \}
  \] where attribute sequence $\lex \in L$, lower bound $a$ and upper bound $b$ are tuples in $\mathcal D$,  respectively.
\end{definition}

\subsection{Constructing Bounds for Range Searches}
Each range search contains two symbolic bounds $a$ and $b$ in the range searches predicate as well 
as an index $\ell$. The primitive searches in a primitive nested loop join
may not specify all attributes in their search predicate. Therefore, the
construction of the lower and upper bound require care. Unspecified 
values need to be padded with infima and suprema values for lower and upper bounds,
respectively.
 We define an unspecified elements for the 
lower/upper bound construction by an artificial constants\footnote{We assume that $\triangle$ is not
  element of any of the domains $\rdom_i$.} $\triangle$.
 We define a bijective index mapping function $i_k: \{1, \ldots,
m\} \rightarrow \{1, \ldots, k+1\}$ that maps the specified elements to
their corresponding constant values, and the unspecified elements to
$\triangle$. We further introduce an artificial value $v_{k+1} =
\triangle$ for the primitive search such that for unspecified
attributes the unspecified symbol $\triangle$ is used. The
construction of the lower and upper bound is performed by the
functions $\textit{lb}$ and $\textit{ub}$, respectively,
\begin{align*}
  a &= \textit{lb}(v_{i_1}, \ldots, v_{i_m}) \\
  b &= \textit{ub}(v_{i_1}, \ldots, v_{i_m})
\end{align*}
that replace the unspecified $\triangle$ value to either the infimum
or supremum of the domain $\rdom_i$, respectively. Formally, the
functions are defined as $\textit{lb}(v_1, \ldots, v_m) = (v'_1,
\ldots, v'_m)$ where \[v'_i = \begin{cases} v_i & \mbox{if $v_i \not = \triangle$} \\ \bot_i & \mbox{otherwise} \end{cases}\]
and  $\textit{ub}(v_1, \ldots, v_m) = (v'_1,\ldots, v'_m)$ where 
\[v'_i = \begin{cases} v_i & \mbox{if $v_i \not = \triangle$} \\
  \top_i & \mbox{otherwise} \end{cases}\] Basically, the functions
$\textit{lb}$ and $\textit{ub}$ are identity functions except for the
case of unspecified values $\triangle$, which are either converted to
infima of suprema of the corresponding element domain. The construction of lower and
upper bounds for partial attribute searches is correct for the full
attribute search, i.e., the are no unspecified values in the value.

\subsection{Computing Index Sets}
The next step is to compute a set of indices which are mapped to range searches predicates. For this there are several 
options available, including producing an index with an orderings given by the default order the attributes syntactically appear in 
the atom, randomized orders etc. However, when complex sets of queries (complex access patterns, variables bindings etc.) 
are present and when large relations are processed, constructing an optimal set of $\ell$s for range searches is 
crucial. As previously states this is is a variation of the classical index selection problem which we 
examine in more detail in Section~\ref{sec:mosp} and demonstrate its performance impact in Section~\ref{sec:experiments}.

\subsection{Range Search Cover}
An important characteristic of range searches is that while exhibiting better worst-case search 
performance, they retain the semantics of primitive searches. We refer to this as a 
To establish this property we  define the notion of a \emph{prefix set} $ \textbf{prefix}_{k}(\lex)$. The prefix 
set produces the first $k$ elements of an index $\lex$ over the set of attributes $\{a_1,\ldots, a_p\}$ of a relation.

\begin{definition}[Prefix Set]
\label{def:prefixset}
Let $\ell=a_{{1}}\prec a_{{2}} \prec \ldots \prec a_{k} \prec a_{k+1} \prec \ldots \prec a_{{p}}$ be a lexicographical order sequence (index), the prefix set of $\ell$ is defined as:

\begin{align*}
\textbf{prefix}_{k}(a_{1}\prec a_{2} & \prec \ldots \prec a_{k} \prec a_{k+1} \prec \ldots \prec a_{p}) = \\
& \begin{cases}
\{a_{1},a_{2},\ldots,a_{k}\},  \mbox{if \ensuremath{k\leq p}},\\
\{a_{1},a_{2},\ldots,a_{p}\},  \mbox{otherwise.}
\end{cases}
\end{align*}

\end{definition}

We say a range search $\sigma_{\rho(\ell, a , b)}(R)$ \emph{covers} a 
primitive search $\searchquery_{x_1=v_1, \ldots, x_k=v_k}(R)$, 
if the $k$-th prefix of $\lex$ result in set $\{x_1, \ldots, x_k\}$. Hence, 
an index represented by an attribute sequence $\lex$ may cover a 
multitude of primitive searches
assuming the elements of its prefixes coincide with the attributes of
the searches. Conversely, if the elements of a primitive search does not
show in the elements of the $k$-th prefix of $\lex$, the primitive
search cannot be covered/executed by a range search using $\lex$.

\begin{lemma}[Range Search Cover]
\label{lemma:rangequerycover}
 \[\forall
  R\subseteq {\mathcal D} : \searchquery_{x_1=v_1, \ldots, x_k=v_k}(R) =
  \sigma_{\rangequery(\lex,a,b)}(R)\], $ \{x_1, \ldots, x_k\} = \textbf{prefix}_{k} (\lex) $,  $a= \textit{lb}(v_{1}, \ldots, v_{k})$, and $b=
  \textit{ub}(v_{1}, \ldots, v_{k})$.
\end{lemma}
% where $\lex=x_1 \prec \ldots \prec x_k$

The correspondence can be extended for any permutation over the set
$\{x_1, \ \ldots, \ x_k\}$. This follows from the commutativity
property of the search predicate, i.e., the actual order of the
equality condition in the search primitive is irrelevant. An extension of 
$\lex$ with further attributes still preserves the
correspondence, i.e., for all lexicographical orders $\lex=x_1 \prec \ldots \prec x_k
  \prec x_{k+1} \prec \ldots \prec x_{k+s} $, it still holds that $\forall
  R\subseteq {\mathcal D} : \searchquery_{x_1=v_1, \ldots, x_k=v_k}(R) =
  \sigma_{\rangequery(\lex,a,b)}(R)$. 
To construct a range search for a primitive search, we need to compute the bounds and lexicographical order.

\begin{example}[Motivating (Cont.)]
Let us assume  we have a set of primitive search predicates extracted from a query
(See Section~\ref{sec:integrate}). The primitive search predicates are 
depicted in the left column, indices in the middle column and associated bounds of the range search predicate 
in the right column in the Table~\ref{tbl:ranget}.

\begin{table*}
\begin{center}
\footnotesize
  \begin{tabular}{l|l|l|l}
  & \multicolumn{3}{c}{Range Predicate $\rangequery_{(\ell,a,b)}$} \\
Primitive Predicate & $\ell$  & $a$ &  $b$ \\
  \hline 
$x = t_1(z)$                   & $x$                 & $\langle t_1(z), \bot, \bot \rangle $  & $\langle t_1(z),\top,\top \rangle $  \\
$x = t_1(y), y = t_1(z)$       & $x \prec y$         & $\langle t_1(y),t_1(z),\bot\rangle$ & $\langle t_1(y),t_1(z),\top\rangle$ \\
$x = t_1(y), z = t_1(z)$          & $x \prec z$         & $\langle t_1(y),\bot,t_1(z)\rangle$ & $\langle t_1(y),\top,t_1(z)\rangle$ \\
$x = t_1(z), y = t_1(y), z = t_1(x)$ & $x \prec y \prec z$ & $\langle t_1(z),t_1(y),t_1(x)\rangle$ & $\langle t_1(z),t_1(y),t_1(x)\rangle$ \\ 
 \end{tabular}
\end{center}
\caption{Range Search Predicate for Primitive Search Predicates\label{tbl:ranget} }
\end{table*}

Here, $c_x$, $c_y$ and $c_z$ denote arbitrary constants, and $\bot$ and $\top$ denote infima and suprima the domain of the 
relation since they are not considered by the lexicographical order $\sqsubseteq_\ell$. 

Using range searches, the complexity of the input query substantially reduces from $\bigo(n^5)$ to \linebreak
$\bigo\left(n \log^4 n \prod_{2 \leq i \leq 5} |\searchquery_{\varphi_i}(A)|\right)$ assuming that results of 
the primitive searches, i.e., $|\searchquery_{\varphi_i}(A)|$ are significantly smaller than $\bigo(n)$. In this 
example, we have to maintain four different indices causing significant overheads for large instances of relation $A$. 
\end{example}

\section{Computing Minimal Index Sets}
\label{sec:mosp}
We have seen that range searches are essential for the efficient execution of Datalog queries. However, when 
constructung range searches, the question remains: \emph{what is the minimal set of indices needed to cover all primitive searches} for
a given relation. 

\subsection{Minimal Order Selection Problem}
Before we define the problem of finding minimal indices, we establish some
additional definitions: let $A=\{x_{1},\ldots,x_{k}\}$ be a finite set of attributes from a given 
relation such that $A \subseteq A_R$. A primitive search 
$\searchquery_{x_1=v_1, \ldots, x_k=v_k}$ is abstracted as a set of search attributes, 
which we refer to as a \emph{search} denoted by $\search \subseteq A$ i.e., $\search = \{x_1, \ldots, x_k \}$. 

As before, we denote an index as $\ell$. The sequence $\lex$ is formed by a subset of 
attributes $\ell \in \mathcal{L} = \bigcup_{X \subseteq A, X\not=\emptyset} \perm (X).$ i.e., $\mathcal{L}$ 
represents the set of all possible permutations/sequences that may be formed by the elements of set 
$X=\{x_1, \ldots, x_k\}$. The set of index sets is defined by $\mathbb{L}$ ranged over by $L$. 

Given a set of searches $\searchset$ for a relation $\rel$ we would like to know which set of indices $L$ will cover 
$\searchset$. We formalize this via the \textbf{l-cover} predicate.

\begin{definition}[l-cover]
We define a predicate $\textbf{l-cover}_{\searchset}(L)$ such that:
\[
\textbf{l-cover}_{\searchset}(L)=\forall \search\in \searchset:\exists \ell \in L:\textbf{prefix}_{|\search|}(\ell)=\search.
\]
\end{definition}

The predicate $\textbf{l-cover}$ provides a means to express the problem of finding the minimal set of indices 
for a relation and its searches which we name the Minimal Order Selection Problem (MOSP).
An input instance of MOSP is given by a set of searches $\searchset$. The set of attributes 
are the attributes of the searches, i.e., $A=\bigcup_{\search \in \searchset} \search$ 
which are relevant for the index selection. 
MOSP seeks to find the set of all solutions where a solution is all minimal sets of 
indices $L$ such that $\textbf{l-cover}_{\searchset}(L)$ holds. 

\begin{definition}[Minimum Order Selection Problem]
The minimum order selection problem finds index sets with minimal
cardinality such that searches are covered by each index set, i.e.,
\[
f_{\searchset}=\arg\min_{L\in\mathbb{L}:\textbf{l-cover}_{\searchset}(L)}|L|.
\]
\end{definition}

\begin{example}
Consider the primitive searches with predicates
$x=v_1$, $x=v_2 \wedge y=v_3$, $x=v_4 \wedge z=v_5$, and 
$x=v_5 \wedge y=v_6 \wedge z=v_7$ over a relation $R$ 
with attribute set $A=\{x,y,z\}$. The actual values 
$v_1$ to $v_7$ in the conditions of the primitive search
are irrelevant as primitive searches are reduced to their
searches (attribute sets) for MOSP, i.e., $\searchset=\langle \{x\},\{x,y\},\{x,z\},\{x,y,z\} \rangle$. One
possible solution for the given instance $(\searchset,A)$ of MOSP would be the index set
$L=\{ x, x \prec y, x \prec z, x \prec y \prec z \}$.  Here each search
is covered by an index. For example, the search $\{x\}$ is covered by
the lexicographical order $x$ and so forth. However, the index set 
$L$  is not a minimal set. For example, the order $x \prec y \prec z$
would cover the searches $\{x\}$, $\{x,y\}$, and $\{x,y,z\}$ since
$x$ is a prefix of length one, $x$, and $y$ are a prefix of length two, 
and $x$,$y$, and $z$ are a prefix of length 3. 
\end{example}

\subsection{Inviability of a Brute-force MOSP Algorithm}
\label{subsec:bfalgorithm}
Before solving MOSP, we would like to understand the size of the solution space of
MOSP. If the number of solutions for an instance of MOSP is very large, a
brute-force algorithm is (assuming a small number of attributes per
input relation) not viable, particularly for high performance engines. Therefore, we 
find bounds on the number of ordered subsets of attribute
set. Constructing a closed form for the cardinality of all possible
lexical graphical orders is hard, however, it can be
bounded.

\begin{lemma}
\label{lem:cardinality}
The cardinality of the set of all sequences is bounded by
\begin{equation*}
m! \leq | \lexset | \leq \euler \cdot m!
\end{equation*}
\end{lemma}

Let $m$ be the number of attributes in a attribute set. For a large $m$, the absolute error of the over-approximation will be
small since the term $m!\sum_{i \geq m} \frac{1}{i!}$ of the
over-approximation will converge quickly. For $m$ between $1$ and $9$,
the values of $|\lexset |$ and the relative error
$\varepsilon=\frac{\euler m! - |\lexset|}{|\lexset|}$ of the
over-approximation is given in the table below:
\begin{center}
\begin{tabular}{r|r|r}
\multicolumn{1}{c|}{$\tl$} & 
\multicolumn{1}{c|}{${|\lexset|} $} &
\multicolumn{1}{c}{$ \% \varepsilon$} \\
\hline
1 &          1 & 171.828 \\
2 &          4 & 35.914 \\
3 &         15 & 8.731 \\
4 &         64 & 1.936 \\
5 &        325 & 0.367 \\
6 &       1956 & 0.059 \\
7 &      13699 & 0.008 \\
8 &     109600 & 0.001 \\
9 &     986409 & $\approx$ 0.000
\end{tabular}
\end{center}

MOSP searches for the smallest subset of $\lexset$ that covers all
primitive searches of the input query. A brute-force approach would
require to find a set of lexicographical orders in search space
$2^\lexset$ for the minimal set of lexicographical orders, i.e.,
$|2^{\lexset}| = 2^{|\lexset|}$. Using the approximation of set
$\lexset$, we obtain a complexity of $\bigo(2^{\euler \cdot m!})$.

\begin{theorem}[MOSP Worst-Case Run-Time]
A brute-force algorithm for MOSP exhibits a worst-case run-time
complexity of $\bigo(2^{m^m})$.
\end{theorem}

The theorem can be shown by using Sterling's approximation. The 
approximation becomes more precise for a large $\tl$.
Note, that a brute-force approach becomes intractable very quickly. Assume that
relation $\rel$ has four attributes. For a relation with 4 attributes, 
a brute-force MOSP algorithm has to test $2^{64}\approx1.8E18$ 
different index subsets for coverage and minimality.

\subsection{Minimal Query Chain Covers}
\label{subsec:algorithm}

In this subsection we present a problem related to MOSP, namely, the Minimum Chain
Cover problem (MCCP) introduced by Dilworth~\cite{Dilworth50}. We have seen that 
a brute force exploration of the MOSP space is infeasible. In the preceding subsections 
we use a combinatorial relationship between the search attributes of primitive searches (present in MCCP) and 
indices (present in MOSP) and exploit the relationship between the two to derive minimal index sets for a 
relation in polynomial time.

A search chain $c\in\mathcal{C}$ is a set of searches $\{\search_{1},\ldots, \search_{k}\}$
that subsume each other and form a total order, i.e., 
$c\equiv \search_{1}\subset \search_{2}\subset\ldots\subset \search_{k}$. We define 
a set of chains $C\in\mathbb{C}$. A
chain $c\in\mathcal{C}$ is a set of searches $\{\search_{1},\ldots, \search_{k}\}$
that subsume each other and form a total order, i.e., $c\equiv \search_{1}\subset \search_{2}\subset\ldots\subset \search_{k}$.
A chain $c$ covers a search $\search$ if $\search\in c$. A set of chains $C\in\mathbb{C}$
cover a search set $\searchset$ if there exists at least one chain $c\in C$
for each search $\search\in \searchset$:

\begin{definition}[c-cover]
We define a predicate $\textbf{c-cover}_{\searchset}(C)$ such that:
\[
\textbf{c-cover}_{\searchset}(C)=\forall \search\in \searchset:\exists c\in C:\search\in c
\]
\end{definition}

The objective of the minimum chain cover problem is to find the smallest
set of chains that cover all searches in $\searchset$, i.e., 

\begin{definition}[Minimum Chain Cover Problem (MCCP)]
\label{def:mccp}
\[
g_{\searchset}=\arg\min_{C\in\mathbb{C}:\textbf{c-cover}_{\searchset}(C)}|C|
\]
\end{definition}

Dilworth's Theorem~\cite{Dilworth50} states that in a finite partial
order, the size of a maximum anti-chain is equal to the minimum number
of chains needed to cover its elements. An anti-chain is a subset of a
partial ordered set such that any two elements in the subset are
unrelated, and a chain is a totally ordered subset of a partial
ordered set.
Although Dilworth's Theorem is non-constructive, there exists
constructive versions that solve the minimum chain problem
either via the maximum matching problem in a bi-partite
graph~\cite{Fulkerson56} or via a max-flow problem~\cite{Pijls2013}. 
Both problems are optimally solvable in polynomial time. 

\subsection{Relationship Between MOSP and MCCP}

The relationship between MCCP and MOSP summarized in 
Fig.~\ref{fig:relationship} and Fig.~\ref{fig:chainindexmap}. In this section we outline the main theorems and lemmata, 
providing full proofs in the appendix. We use the relationship between MOSP and MCCP to develop an algorithm 
for solving MOSP as outlined in Subsection~\ref{subsec:algorithm}. Our approach defines two mapping functions 
which contain several properties, which we use to translate solutions from one space to another.

\subsubsection{Mapping Functions} 
We first define mappings between indices and chains $\alpha$. This mapping is defined on 
two levels as follows:

\begin{definition}[Index to Chain Mapping $\alpha$]
\label{def:alpha}
\begin{align*}
&\alpha_{\searchset}^{0}:\mathcal{L}\rightarrow\mathcal{C} \textit{ s.t. } \ell \mapsto\{\search\in \searchset \ \mid \ \textbf{prefix}_{|\search|}(\ell)=\search\}\\
&\alpha_{\searchset}^{1}:\mathbb{L}\rightarrow\mathbb{C} \textit{ s.t. } L\mapsto\{\alpha_{\searchset}^{0}(\ell) \ \mid \ \ell\in L\}
\end{align*}
\end{definition}

We highlight that by the definition of the prefix set function (Def.~\ref{def:prefixset} in 
Section~\ref{sec:cover}), the searches
$\{\search_{1},\ldots,\search_{k}\}$ of $\alpha_{\searchset}^0(\ell)$ form a chain 
$c\equiv \search_{1}\subset \search_{2}\subset\ldots\subset \search_{k}$. Similarity, 
we also define a mapping $\gamma$ between chains and indices as follows:

%We denote $\mathcal{C}$ as the set of all chains for a given search
%set $Q$. 

\begin{definition}[Chain to Index Mapping $\gamma$]
\label{def:gamma}
\begin{align*}
\gamma_{\searchset}^{0}:&\mathcal{C}\rightarrow\mathcal{\mathbb{L}} \textit{ s.t. } \search_{1}\subset \search_{2}\subset\ldots\subset \search_{k-1}\subset \search_{k}\mapsto \\
&\langle \search_{1}\prec \search_{2}-\search_{1}\prec\ldots\prec \search_{k}-\search_{k-1}\rangle \\
\gamma_{\searchset}^{1}:&\mathbb{C}\rightarrow2^{\mathbb{L}} \textit{ s.t. } \{c_{1},\ldots,c_{k}\}\mapsto \\
&\left\{ \{\ell_{1},\ldots,\ell_{k}\}|\ell_{1}\in\gamma_{\searchset}^{0}(c_{1}),\ldots,\ell_{k}\in\gamma_{\searchset}^{0}(c_{k})\right\}
\end{align*}
\end{definition}

We observe that for all chains $c\in\mathcal{C}\setminus\{\emptyset\}$
that contain at least one search, there exists at least one index in
$\gamma_{\searchset}^{0}(c)$. 

%\begin{observation}
%\[
%\forall c\in\mathcal{C}\setminus\{\emptyset\}:|\gamma_{Q}^{0}(c)|>0
%\]
%\end{observation}

%\begin{lemma}
%\label{lemma:subsumption}
%\[
%\forall c\in\mathcal{C}:\forall \ell \in\gamma_{Q}^{0}(c): c \subseteq \alpha_{Q}^{0}(l)
%\]
%\end{lemma}

%\begin{definition}
%Function $\gamma_{Q}^{1}:\mathbb{C}\rightarrow2^{\mathbb{L}}$ generates
%a set of index sets for a set of chains, i.e., 
%\[
%\{c_{1},\ldots,c_{k}\}\mapsto\left\{ \{l_{1},\ldots,l_{k}\}|l_{1}\in\gamma_{Q}^{0}(c_{1}),\ldots,l_{k}\in\gamma_{Q}^{0}(c_{k})\right\} .
%\]
%\end{definition}

\subsubsection{Cardinality Relationship} 
The first set of lemmata, define properties on the cardinality relationship between 
the lexicographical and chain spaces.

Below we establishe the cardinality relationship between a set of indices and chains 
via the $\gamma$ function.

\begin{lemma}[$\gamma$ and Chain Cardinality]
\label{lemma:chaincard}
We observe that $\forall C\in\mathbb{C}:\forall L\in\gamma_{\searchset}^{1}(C):|L|\leq|C|$
(by construction) and for any non-empty chain set $C$, $|\gamma^{1}_{\searchset}(C)|>0$,
i.e., there exists always at least one index set.
\end{lemma}

In the above lemma we see that the size of each index set mapping to a chain set is bound by the size 
of the chain set. Conversely, we assert in the next Lemma that $\alpha$ does not modify the size of 
a set of lexicographical orders in the MOSP solution space.

\begin{lemma}[$\alpha$ and MOSP Element Cardinality]
\label{lemma:set-size}
The size of a set of lexicographical orders is preserved by $\alpha$, i.e.,  
\[
\forall L\in f_{\searchset}:|\alpha_{\searchset}^{1}(L)|=|L|
\]
\end{lemma}

\subsubsection{Cover Relationship} 
Another important set of lemmata are ones that reason about covers and precedence of covers between 
the index and chain spaces. The first lemma states that $\alpha$ preserves covers in both spaces: 

\begin{lemma}[$\alpha$ and Cover Equivalence]
\label{lemma:cover}
The $\alpha$ of all lexicographical sets that are \textbf{l-covers} are \textbf{c-covers}, i.e., 
\[
\forall L\in\mathbb{L}:\textbf{l-cover}_{\searchset}(L)\Leftrightarrow\textbf{c-cover}_{\searchset}(\alpha_{\searchset}^{1}(L))
\]
\end{lemma}

On the other hand, the lemma for $\gamma$ is weaker. The cover property is 
preserved only from \textbf{c-covers} to \textbf{l-covers}, not the other way around. We conjecture that both directions 
hold for minimal chain covers, however, we do not prove this as it is not required for our 
approach.

\begin{lemma}[$\gamma$ and Cover Implication]
\label{lemma:gamma-cover}
If a chain set is a c-cover then all lexicographical order sets in the $\gamma$ of the chain cover 
is a \textbf{l-cover}, i.e.,
\[
\forall C\in\mathbb{C}:\textbf{c-cover}_{\searchset}(C)\Rightarrow\forall L\in\gamma_{\searchset}^{1}(C):\textbf{l-cover}_{\searchset}(L)
\]
\end{lemma}

\subsubsection{Minimum Cover Relationship} 
Given the established relationships between cardinality and covers, we can state the following 
minimal cover theorems.

\begin{theorem}[Solution Preservation of $\alpha$]
\label{thm:alpha-optimal}
The $\alpha$ of all optimal lexicographical orders is a optimal chain cover, i.e., 
\[
\forall L\in f_{\searchset}:\alpha^{1}(L)\in g_{\searchset}
\]
\end{theorem}

\begin{theorem}[Solution Preservation of $\gamma$]
\label{thm:main2}
For all optimal chain covers there exists an optimal lexicographical order set that is 
optimal i.e.,
\[
\forall C\in g_{\searchset}:\forall L\in\gamma_{\searchset}^{1}(C):L\in f_{\searchset}
\]
\end{theorem}

\begin{figure}[h]
\begin{minipage}{0.47\textwidth}
\centering
\centering
\includegraphics[width=0.88\textwidth]{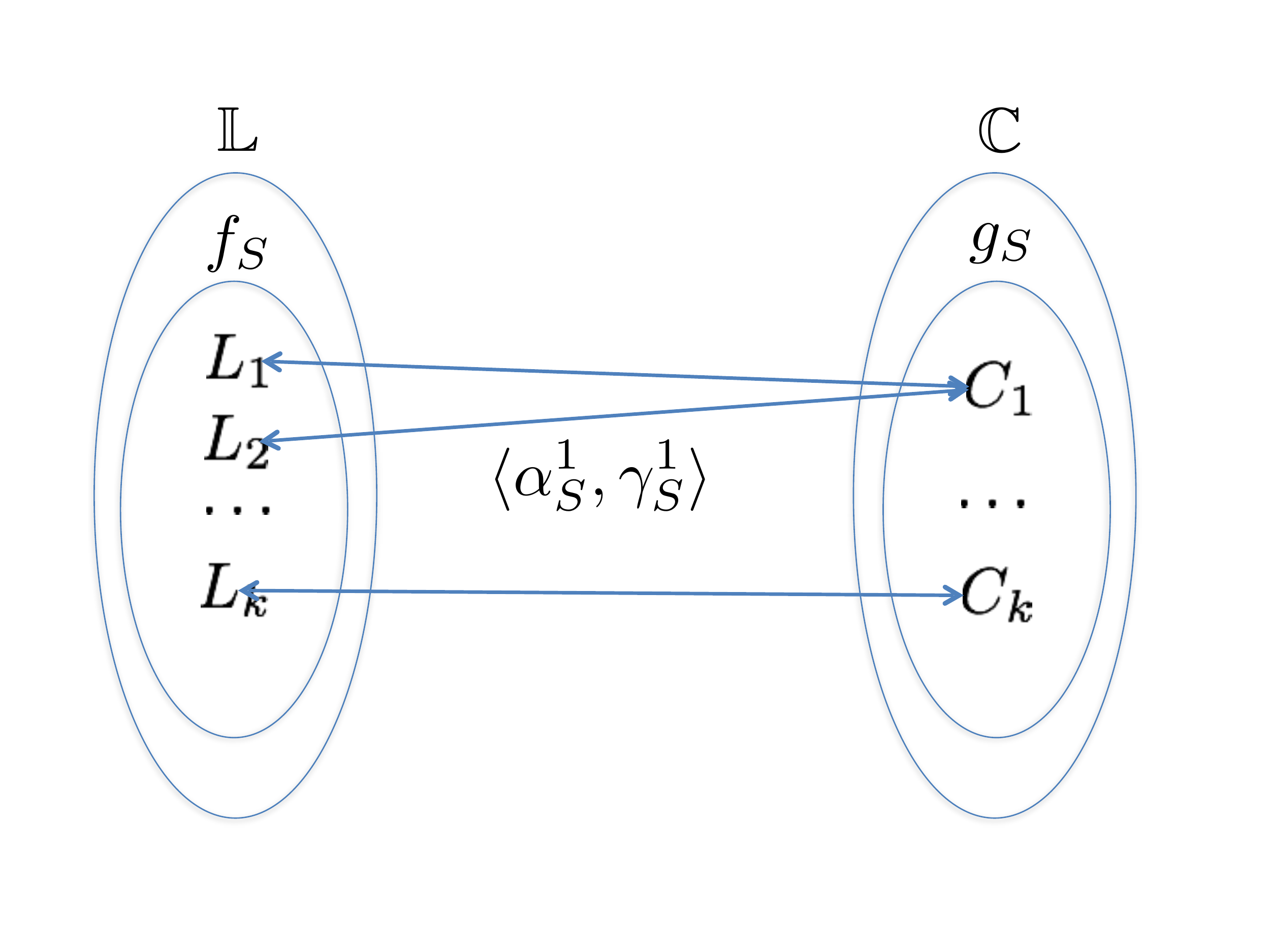}
\caption{Minimum Cover Mapping Between $f_\searchset$ and $g_\searchset$\label{fig:relationship}}
\end{minipage}
\begin{minipage}{0.47\textwidth}
\centering
\centering
\includegraphics[width=0.92\textwidth]{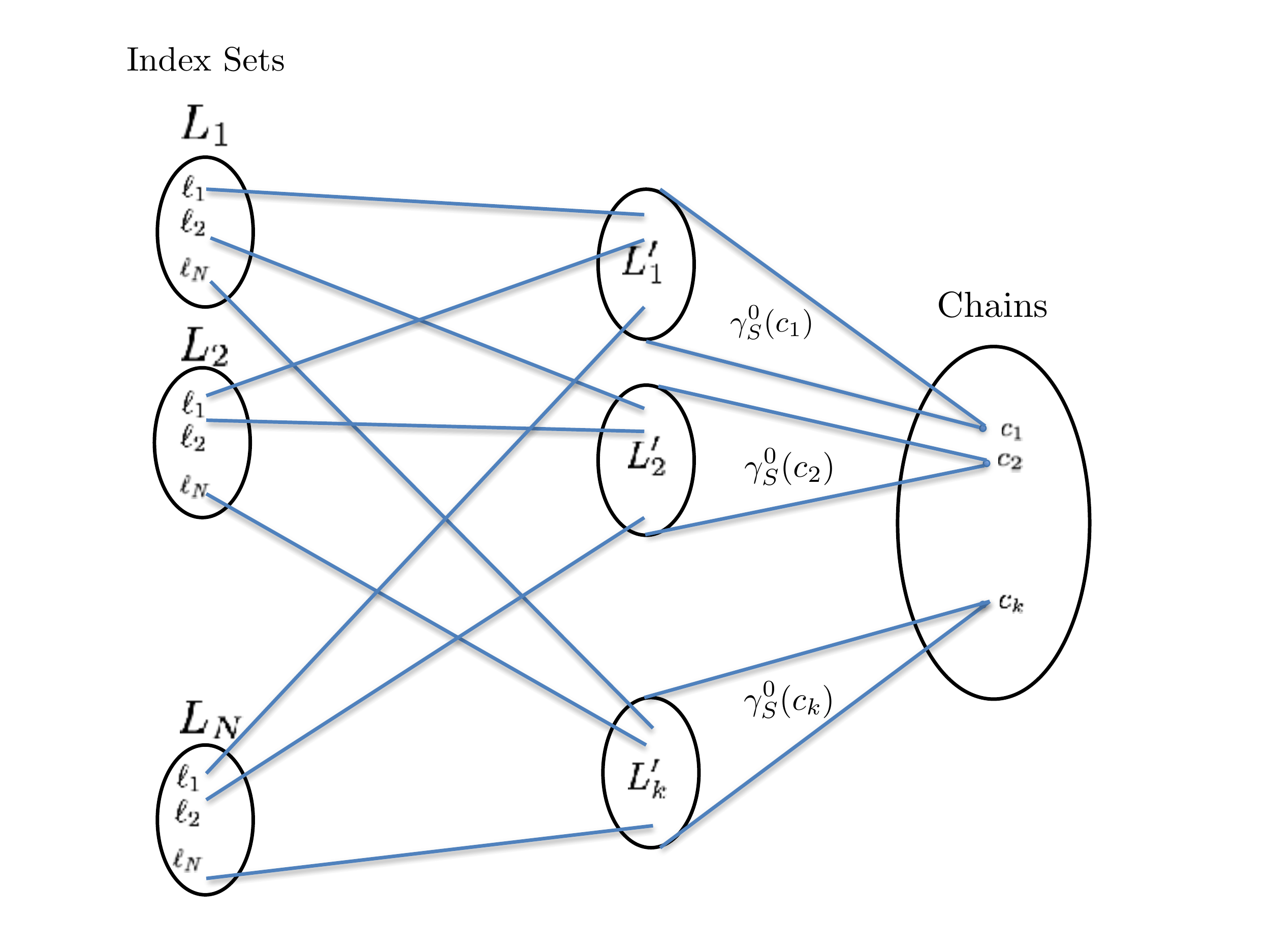}
\caption{Mapping Between Chains and Index Sets\label{fig:chainindexmap}}
\end{minipage}
\end{figure}

The two theorems above establish a clear relationship between the two spaces of solutions. We exploit this 
in our MOSP algorithm.

\subsection{An Optimal MOSP Algorithm}
\label{subsec:algorithm}
To practically apply the theorms in the previous section we introduce a new algorithm that finds 
a minimal set of indices for an instance $(\searchset,A)$ of MOSP. The algorithm follows 
from Theorem~\ref{thm:main2}. However, as we only need a single minimum chain cover, 
our $\gamma_{\searchset}^1$ essentially performs a choice/selection of any given minimum chain cover 
and converts to to a minimum index cover.

We outline the algorithm based on~\cite{Fulkerson56} of finding an optimal solution of MCCP in
Algorithm~\ref{alg:chain-cover} for a search set $\searchset$. First, a bi-partite graph is
constructed whose vertex sets are the search sets in both partitions of
the bi-partite graph (cf.~Line~1). Second, an edge between two searches
$\search$ and $\search'$ is constructed if $\search$ is a strict subset of $\search'$ . The
maximum matching algorithm computes the matching set $E$. Chains are
constructed from the matching set by finding the searches that start a
chain, i.e., are the smallest element of a chained and do not have a a
predecessor. In the algorithm such elements are found in line~4. The
smallest element connects via edges immediately and intermediately all
searches in the chain. The chain is added to the chain set $C$ in
line~6.

Algorithm~\ref{alg:auto-index} converts the chains to an index sets. However, due to the 
relationship summarized in Fig.~\ref{fig:chainindexmap} a single chain may produce several indices. To 
simply the algorithm and due to only for our problem requiring a single index set we 
arbitrarily choose a single index from a chain with $\textbf{Choose}$ operation.

\SetKwProg{Fn}{Function}{\string:}{}
\begin{algorithm}[t]
\Fn{\textbf{MinChainCover}(\searchset)}{
 \KwData{Set of Searches $\searchset$}
 \KwResult{Minimum Chain Cover $\mathcal C$}
  $E := \textbf{MaxMatching}(\searchset,\searchset, \{ (\search,\search') \in \searchset \times \searchset \mid \search \subset \search' \})$\;
 initialize ${\mathcal C}$ to  the empty set\;
 \ForAll{$u_1\in \searchset: \not \exists (u_0,u_1) \in E$}{
   {
      find max.~set $\{(u_1,u_2), (u_2,u_3), \ldots, (u_{k-1},u_k)\}  \subseteq E$ \;
      add $u_1 \subset u_2 \subset u_3 \ldots \subset u_{k-1} \subset u_k $ to ${\mathcal C}$ \;
  }
}
\Return{$\mathcal C$}
}
\caption{Find Minimum Chain Cover for Searches\label{alg:chain-cover}}
\end{algorithm}

\begin{algorithm}[t]
 \Fn{\textbf{MinIndex}(\searchset)}{
   \KwData{Set of Searches $\searchset$}
   \KwResult{Set of indices to cover searches}
 ${\mathcal C} := \textbf{MinChainCover}(\searchset)$\;
 initialize $L$ to the empty set\;
 \ForAll{$\search_1 \subset \search_2 \ldots \subset \search_{k-1} \subset \search_k   \in {\mathcal C}$}{
   add $\textbf{Choose}\left( \search_1 \prec \search_2 - \search_1 \prec  \ldots \prec \search_{k} - \search_{k-1}\right)$  to $L$\;
} 
\Return{$L$}}
\caption{Compute Minimal Index Set \label{alg:auto-index}}
\end{algorithm}

\pgfdeclarelayer{background}
\pgfdeclarelayer{foreground}
\pgfsetlayers{background,main,foreground}

\begin{figure*}[t]
\footnotesize
\centering
\subfigure[Matching Problem]{
\begin{tikzpicture}[semithick,
  every node/.style={draw,circle},
  fsnode/.style={fill=myblue},
  ssnode/.style={fill=mygreen},
  every fit/.style={ellipse,draw,inner sep=-2pt,text width=2cm},
  -,shorten >= 3pt,shorten <= 3pt
]
  % the vertices of U
\begin{scope}[start chain=going below,node distance=7mm]
  \foreach \i/\t in {1/\{x\},2/\{x\com y\},3/\{x \com z\},4/\{x\com y \com z\}}
           { \node[fsnode,on chain] (f\i)
                [label={[left,xshift=0mm,yshift=-2.0mm]{\footnotesize $\t$}}] {}; }
\end{scope}
% the vertices of V
\begin{scope}[xshift=3.5cm,start chain=going below,node distance=7mm]
\foreach \i/\t in {1/\{x\},2/\{x\com y\},3/\{x \com z\},4/\{x\com y \com z\}} {
  \node[ssnode,on chain] (s\i) [label={[right,xshift=0mm,yshift=-2.0mm]{\footnotesize $\t$}}] {};
}
\end{scope}
% the set U
\node [myblue,fit=(f1) (f4),label=above:$\searchset$] {};
% the set V
\node [mygreen,fit=(s1) (s4),label=above:$\searchset$] {};
% the edges
\draw [line width=0.4mm] (f1) -- (s2);
\draw [line width=0.4mm] (f1) -- (s3);
\draw [line width=0.4mm] (f1) -- (s4);
\draw [line width=0.4mm] (f2) -- (s4);
\draw [line width=0.4mm] (f3) -- (s4);
\end{tikzpicture}
\label{fig:example-problem}
}
\subfigure[Matching Solution]{
\begin{tikzpicture}[thick,
  every node/.style={draw,circle},
  fsnode/.style={fill=myblue},
  ssnode/.style={fill=mygreen},
  every fit/.style={ellipse,draw,inner sep=-2pt,text width=2cm},
  -,shorten >= 3pt,shorten <= 3pt
]
  % the vertices of U
\begin{scope}[start chain=going below,node distance=7mm]
  \foreach \i/\t in {1/\{x\},2/\{x\com y\},3/\{x \com z\},4/\{x\com y \com z\}}
           { \node[fsnode,on chain] (f\i)
                [label={[left,xshift=0mm,yshift=-2.0mm]{\footnotesize $\t$}}] {}; }
\end{scope}
% the vertices of V
\begin{scope}[xshift=3.5cm,start chain=going below,node distance=7mm]
\foreach \i/\t in {1/\{x\},2/\{x\com y\},3/\{x \com z\},4/\{x\com y \com z\}} {
  \node[ssnode,on chain] (s\i) [label={[right,xshift=0mm,yshift=-2.0mm]{\footnotesize $\t$}}] {};
}
\end{scope}
% the set U
\node [myblue,fit=(f1) (f4),label=above:$\searchset$] {};
% the set V
\node [mygreen,fit=(s1) (s4),label=above:$\searchset$] {};
% the edges
\draw [line width = 0.4mm, brown] (f1) -- (s2);
\draw [line width = 0.4mm, brown] (f2) -- (s4);
%\begin{pgfonlayer}{background}
%        \draw[rounded corners=1.5em,line width=1.5em,gray!20,cap=round]
%                (f1.center) -- (s2.east) -- (f2.west) -- (s4.center);
%\end{pgfonlayer}
\end{tikzpicture}
\label{fig:example-solution}
}
\subfigure[Chain Cover]{
\begin{tikzpicture}[thick,
  every node/.style={draw,circle},
  fsnode/.style={fill=myred},
  every fit/.style={ellipse,draw,inner sep=-2pt,text width=2cm},
  -,shorten >= 3pt,shorten <= 3pt
]
\node[fsnode] at (2,3.6)(n1)
       [label={[left,xshift=+1mm,yshift=2mm]{\footnotesize $\{x\}$}}] {};
\node[fsnode] at (1.4,1.8)(n2)
       [label={[left,xshift=+1mm,yshift=2mm]{\footnotesize $\{x,y\}$}}] {};
\node[fsnode] at (2.6,1.8)(n3)
       [label={[right,xshift=-1mm,yshift=2mm]{\footnotesize $\{x,z\}$}}] {};
\node[fsnode] at (2,0)(n4)
       [label={[left,xshift=+1mm,yshift=-5mm]{\footnotesize $\{x,y,z\}$}}] {};
% the edges
\draw [line width = 0.4mm] (n1) -- (n2);
\draw [line width = 0.4mm] (n1) -- (n3);
\draw [line width = 0.4mm] (n2) -- (n4);
\draw [line width = 0.4mm] (n3) -- (n4);
\begin{pgfonlayer}{background}
        \draw[rounded corners=2.4em,line width=2.4em,blue!20,cap=round]
                (n1.center) -- (n2.west) -- (n4.center);
\end{pgfonlayer}
\begin{pgfonlayer}{background}
        \draw[rounded corners=2.4em,line width=2.4em,blue!20,cap=round]
                (n3.north) -- (n3.south);
\end{pgfonlayer}
\end{tikzpicture}
\label{fig:example-chaincover}
}
\caption{Motivating Example of Fulkerson's Maximum Matching Reduction
  for Dilworth's Theorem. 
  Partial ordered set is the set of searches
  $\{x\}$, $\{x,y\}$, $\{x,z\}$, and $\{x,y,z\}$ of relation
  $A$. Bi-partite construction and solution of maximum matching
  problem induce a minimal chain cover. The chains induce minimal
  number of indices $x \prec y \prec z$ and $x \prec z$. 
}
\label{fig:mpsp}
\end{figure*}
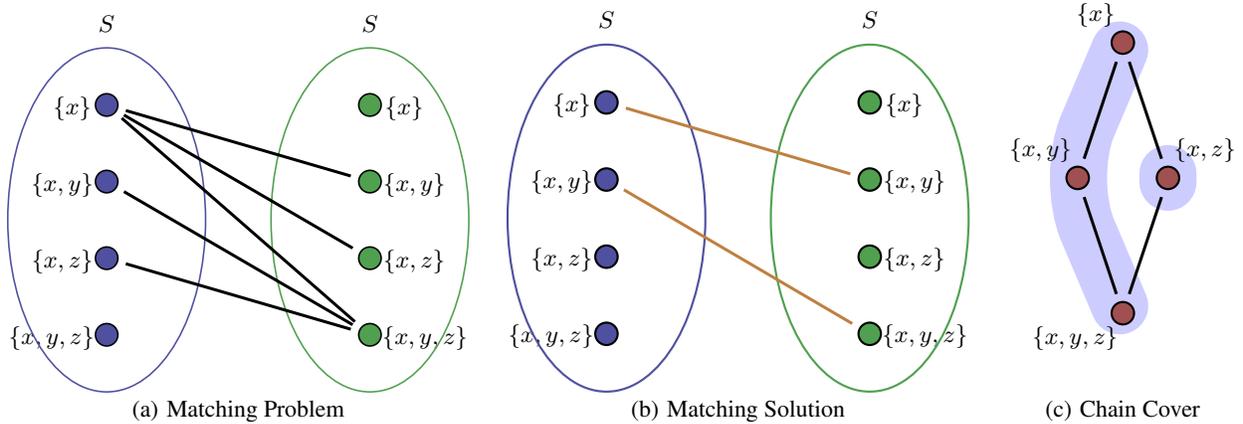

\begin{example}[Motivating Example (Cont.)]
Consider the motivating example that has the following search set,
$
 \searchset=\{ \{x\}, \{x,y\}, \{x,z\}, \{x,y,$ $z\}\}
$
that needs to be covered by the smallest set of indices. First, we
construct a bi-partite graph with nodes in $\searchset$ in both partitions. The
edge set is given by the strict subset relationship between a search
pair, i.e., $ (\{x\}, \{x,y\}), (\{x\}, \{x,z\}), (\{x\}, $ $\{x,y,z\}),
(\{x,y\}, \{x,y,z\})$ and $(\{x,z\}, \{x,y,z\})$.  The bi-partite
graph is depicted in Fig.~\ref{fig:example-problem}, and the
matching set of the maximal matching solution is depicted in
Fig.~\ref{fig:example-solution}. The solution of the maximal
matching algorithm is given by the matching set,
$
E = \{ (\{x\}, \{x,y\}), (\{x,y\}, \{x,y,z\}) \}.
$
With Algorithm~\ref{alg:chain-cover} we obtain a chain cover $C$
containing the following chains,
$ \{x\} \subset \{x,y\} \subset \{x,y,z\} $ and $ \{x,z\}$ 
that is depicted in Fig.~\ref{fig:example-chaincover}. It is apparent
that the two chains are minimal for the cover since the cardinality of
the maximum anti-chain (i.e., $\{\{x,y\}, \{x,z\}\}$) is also two
(cf.~Dilworth's Theorem).

The Algorithm~\ref{alg:auto-index} converts the chain cover to
indices using a $\gamma$ transformation. The first chain is converted as follows:
$
   \{ x \} \subset \{x,y\} \subset \{x,y,z\}  $ $ \Rightarrow $ $
  \{ x \} \prec \{x,y\} - \{x\} \prec \{x,y,z\} - \{x,y\} $ $ 
  \Rightarrow x \prec y \prec z 
$
Since the smallest element and
the set difference are singletons, there exists only a single
index that covers the searches of the chain. The
second chain consists of a single element $\{x,z\}$. This chain
induces two possible indices, i.e., $x \prec z$ and $z \prec x$, and
the choice is arbitrary to find an optimal solution for the MOSP
problem. An updated index mapping is defined below:

\begin{center}
\footnotesize
  \begin{tabular}{l|l}
Primitive Predicate & Assigned $\ell$ \\
  \hline 
$x = t_1(z)$                         & $x \prec y \prec z$  \\
$x = t_1(y), y = t_1(z)$             & $x \prec y \prec z$   \\
$x = t_1(y), z = t_1(z)$             & $x \prec z$           \\
$x = t_1(z), y = t_1(y), z = t_1(x)$ & $x \prec y \prec z$  \\ 
 \end{tabular}
\end{center}

\end{example}

\section{Datalog Engine Integration}
\label{sec:integrate}
In this section we describe the practicalities of how our technique is integrated into a Datalog engine, using 
\souffle as an example. The purpose of this section is to allow our method to be replicated in other high performance 
Datalog engines\footnote{And in-memory databases in general}~\cite{Hoder11, logicblox, Bierhoff05, Coral, GlueNail, Socialite}.

\begin{figure*}
\begin{center}
  \begin{tikzpicture}[scale=0.3,
  text1/.style={scale=0.7,draw=none,font=\small},
  block/.style={scale=0.7,ultra thick,draw,fill=gray!20,rectangle,minimum width={2cm}, minimum height={1.5cm}, font=\small},
  block2/.style={scale=0.7,ultra thick,draw,fill=yellow!20,rectangle,minimum width={2cm}, minimum height={1.5cm}, font=\small},
  action/.style={scale=0.7,ultra thick,draw,fill=green!10,ellipse,minimum width={1.5cm}, minimum height={1.7cm}, font=\small}
  ]
    \matrix[row sep=.3cm,column sep = .3cm]{
     & & & & &\node[block2](IO0)  {\begin{tabular}{c} {\bf Index Optimizer}\\for Relation $R_1$\end{tabular}}; & \\     
     \node[action](Q) {Query}; &
     \node[block](QP) {\begin{tabular}{c} {\bf Query}\\{\bf Translator}\end{tabular}}; &
     \node[action](LN0) {\begin{tabular}{c} Loop\\Nest\end{tabular}}; &
     \node[block](R) {\begin{tabular}{c} {\bf Search}\\{\bf Rewriter}\end{tabular}}; &
     \node[action](LN) {\begin{tabular}{c} Primitive \\Loop\end{tabular}}; & \node[text1]{$\ldots$}; &
     \node[action](OLN) {\begin{tabular}{c} Range \\Loop\end{tabular}}; \\
     & & & & &\node[block2](IO2)  {\begin{tabular}{c} {\bf Index Optimizer}\\for Relation $R_l$\end{tabular}}; & \\
    }; 
    \path     (Q) edge[->,line width=1.6pt] (QP) 
      (QP) edge[->,line width=1.6pt] (LN0)
      (LN0) edge[->,line width=1.6pt] (R)
      (R) edge[->,line width=1.6pt] (LN)      
      (LN) edge[->,line width=1.6pt] (IO0)
      (LN) edge[->,line width=1.6pt] (IO2) 
      (IO0) edge[->,line width=1.6pt] (OLN)
      (IO2) edge[->,line width=1.6pt] (OLN); 
\end{tikzpicture}
\end{center}
\caption{Query optimisation pipeline for input relations $R_1, \ldots, R_l$: index optimiser is invoked for each input relations separately. }
\label{fig:q-pipeline}
\end{figure*}
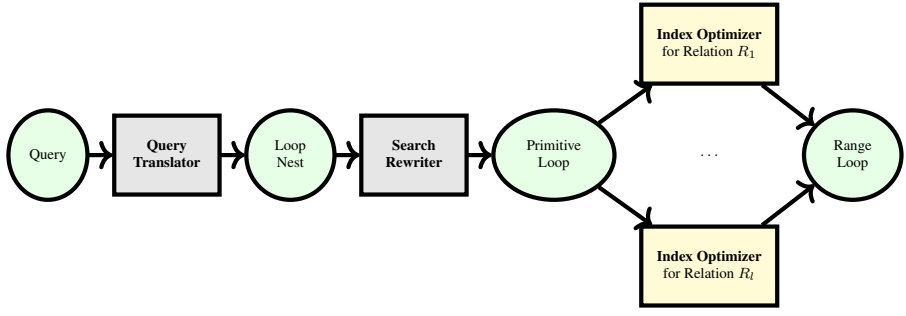

The requirements of our approach (1) that queries for a relational database system are expressed
in a domain specific language, e.g., SQL~\cite{Groff09}, 
Datalog~\cite{Ceri89}, whose underlying query semantics resembles a
relational algebra system~\cite{Codd70,Abiteboul95} employing the 
usual set operators including product, projection, and selection on relations denoted  by $R_1, \ldots, R_l$ 
producing as a result an output relation $R'$. (2) an engine converts joins to a nested loop join, that resembles 
our primitive nested loop joins.

Our approach performs several rewrite transformaitons that we summarize in the pipeline in 
Fig.~\ref{fig:q-pipeline}. In the first step, a \emph{query translator} converts
an input query\footnote{For sake of simplicity we exhibit our approach only for a single query; however 
the approach can be extended to a collection of queries, sub-queries, etc.} to a \emph{nested loop join} 
(also known as \emph{join nested loop join}). A nested loop join represents an executable imperative 
program of the input query constructed by a collection of nested loops. Each loop in the nested loop join enumerates 
tuples of a relation that occur in the input query, and filters tuples according to \emph{loop predicates}. The 
loop body of the most inner loop projects the selected tuples of the loops to a new tuple that will be added 
to the output relation of the query if the tuple does not exists. The nested loop join
is rewritten several times to obtain nested loop joins containing index-operations denoted by 
range nested loop join.

\begin{figure}[ht]
\begin{center}
\begin{adjustbox}{bgcolor=blue!20}
\begin{minipage}[t]{0.7\textwidth}
\begin{tabbing}
\hspace{1.2cm}\=\hspace{0.3cm}\=\hspace{0.3cm}\=\hspace{0.3cm}\=\hspace{0.3cm}\=\hspace{0.3cm}\=\kill 
$\textit{loop}_1$: \> $\textbf{for all} \, t_1 \in \rel_{i_1} : \varphi_1(t_1) \, \textbf{do}$ \\
$\textit{loop}_2$: \>\> $\textbf{for all} \, t_2 \in \rel_{i_2} : \varphi_2(t_1, t_2) \, \textbf{do}$ \\
\>\>\> $\ldots$\\
$\textit{loop}_k$: \>\>\>\> \textbf{for all} $t_k \in \rel_{i_k} : \varphi_k(t_1, t_2, \ldots, t_k)$  \textbf{ do}\\ 
\>\>\>\>\>\textbf{if} $\varphi_k(t_1, \ldots, t_m) \not \in R'$ \textbf{ then} \\
\>\>\>\>\>\>\textbf{add } $\pi(t_1, \ldots, t_m)$ \textbf{ to } $R'$\\
\>\>\>\>\>\textbf{endif}\\
\>\>\>\>\textbf{endfor}\\
\>\>\> $\ldots$\\
\>\>\textbf{endfor}\\
\>\textbf{endfor}
\end{tabbing}
\end{minipage}
\end{adjustbox} 
\end{center}
\caption{Loop-Nest: a relational algebra query is translated by a query planner to nested loop join; each 
loop enumerates tuples of a relation and filters selected tuples $t_1, \ldots, t_j$  using a tuple 
predicate $\varphi_j(t_1, \ldots, t_j)$. The inner most loop projects selected tuples $t_1, \ldots, t_k$ 
to a tuple for the output relation $R'$. If the tuple does not exist, it will be added to the output 
relation.} 
\label{fig:nested loop join}
\end{figure}

The structure of a nested loop join is shown in Fig.~\ref{fig:nested loop join}. The loops of the nested loop join are 
labelled by $\textit{loop}_j$ for all $j$, $1\leq j \leq k$. 
The $j$-th loop selects a tuple $t_j$ from relation $R_{i_j}$ where $i_j$ is used to associate the 
$j$-th loop to one of the input relations $R_1, \ldots, R_l$. The \emph{loop predicate} $\varphi_j(t_1, \ldots, t_j)$ 
is defined over the tuple $t_j$ of the current loop and the selected tuples of the outer loops 
$t_1$, \ldots, $t_{j-1}$. The loop predicate filters the tuples of relation $R_{i_j}$, i.e., only if 
the loop predicate holds for the currently selected tuples $t_1, \ldots, t_{j}$,  the loop body is 
executed with the selected tuples $t_1, \ldots, t_j$.   Note that in some cases, the search predicate 
holds  independent of the elements of the tuples (i.e. represents the true value); for this case, all 
tuples of the relation $R_{i_j}$ are enumerated and the search predicate can be omitted. 

The query translator selects the best loop order, minimising the iteration space of the nested loop join with 
the aid of a query planner~\cite{Abiteboul95} or user hints\footnote{.plan directive in \souffle}. Conditions are 
hoisted to the outer-most loop where they
are still admissible in order to prune the iteration spaces effectively. 
This technique is also referred to as levelling~\cite{SKP15}. After the translation to nested loop joins we 
have two subsequent transformations that converts the nested loop join to a nested loop join using index operations to speed
up the execution of query.

\subsection{Search Rewriter} 
The second step the nested loop join is 
transformed to a nested loop join with 
\emph{primitive searches}, which we refer to as $\searchquery$-nested loop join. 

%A primitive search filters tuples of an input relation $R$ using a \emph{search predicate} that is 
%a conjunction of equality predicates over the attributes of the relation $R$, i.e.,   
%\[
% \searchquery_{x_1=v_1, \ldots, x_k=v_k} (\rel) = \{ t \in \rel \mid
%   t(x_1)=v_1 \wedge \ldots t(x_k) = v_k \}
%\] 
%where $x_1=v_1, \ldots, x_k=v_k$ is the search predicate. 
%A primitive search selects a subset of tuples in relation $\rel$ for which the 
%\emph{search predicate} $x_1=v_1 \wedge \ldots \wedge x_k=v_k$ holds. The values $v_1, \ldots, v_k$
%of the search predicate can be either constants or tuple elements of outer loops. 

In 
a subsequent transformation step, a primitive search will be replaced by an 
index operation on relation $R$. Hence, 
a large number of primitive searches in the nested loop join will make the execution 
of the query more efficient.  The rewriting of the nested loop join to a $\searchquery$-nested loop join is mainly a 
syntactical rewrite step and is shown in Figure~\ref{fig:nested loop join-rewrite}.
\begin{figure}
\begin{center}
\adjustbox{bgcolor=blue!20}{
\begin{minipage}[b]{0.4\textwidth}
\begin{tabbing}
\hspace{1.0cm}\=\hspace{0.2cm}\=\hspace{0.2cm}\=\hspace{0.3cm}\=\hspace{0.3cm}\=\hspace{0.3cm}\=\kill 
$\textit{loop}_j$: \> $\textbf{for all} \, t_j \in \rel_{i_j} : \varphi_j(t_1, \ldots, t_j) \, \textbf{do}$ \\
\>\>\> $\ldots$\\
\>\textbf{endfor}
\end{tabbing}
\end{minipage}
}\\
\vspace{3mm}
$\Downarrow$\\
\vspace{3mm}
\adjustbox{bgcolor=blue!20}{
\begin{minipage}[b]{0.4\textwidth}
\begin{tabbing}
\hspace{1.0cm}\=\hspace{0.2cm}\=\hspace{0.2cm}\=\hspace{0.3cm}\=\hspace{0.3cm}\=\hspace{0.3cm}\=\kill 
$\textit{loop}_j$: \> $\textbf{for all} \, t_j \in \searchquery_{\varphi'_j}(\rel_{i_j}) : \varphi''_1(t_1, \ldots, t_j) \, \textbf{do}$ \\
\>\>\> $\ldots$\\
\>\textbf{endfor}
\end{tabbing}
\end{minipage}
}
\end{center}
\caption{Rewriting of loops in nested loop join using primitive searches: loop predicate $\varphi_j$ of the $j$-th loop is replaced by a primitive search with search predicate $\varphi'_j$ and remaining loop predicate $\varphi''_j$.}
\label{fig:nested loop join-rewrite}
\end{figure}
The primitive nested loop join enumerates tuples via the primitive searches, i.e., 
the original condition $\varphi_j$ of the $j$-th loop is broken up
into a search predicate $\varphi'_j$ consisting of a conjunction of equality predicates along 
with the remaining predicate $\varphi''_j$, i.e., $\varphi_j = \varphi'_j \wedge \varphi''_j$.  
Note that the values of the search predicate $\varphi'_j$ can only be constants or 
tuple elements of outer loops (which would be fixed when executing the $j$-th loop); otherwise it cannot be replaced by an index operation.  
For complex loop predicates there might still be a reduced loop predicate $\varphi''_j$ that needs to be evaluated for all tuples generated by 
the primitive search. 

Each search predicate is replaced by an index operation to
reduce the loop-iteration space further. Note that if no primitive search 
predicate can be identified in $\varphi_j$, the search predicate $\varphi'_j$
becomes empty and holds for all tuples, i.e., $\searchquery_{\textit{true}}(R_{i_j})$.
Such a pathological case of a primitive search can be rewritten again 
to the relation $R_{i_j}$ itself. 
 
\subsection{Index Optimizer} 
The final step the $\searchquery$ nested loop join is converted 
to range searches. In our approach indices are associated to a single relation only,
hence, the index optimisation is performed separately for each input relation. As 
described in previous sections a lexicographical order is required
and the index optimiser chooses the minimal number of lexicographical 
orders. 
\raggedbottom
\begin{example}[Motivating Example (Cont.)]
Recall the Datalog rule from the motivating example.
\begin{align*}
B(r,p,q) \leftarrow & A(r,p,q), A(q,\_,\_), \\ & A(p,q,\_), A(p,\_,q), A(q,p,r).
\end{align*}

%Alternatively, the Datalog rule can be expressed as a SQL query
%as well,
%\begin{verbatim}
%SELECT A1.x,A1,y,A1.z 
%INTO B
%FROM A AS A1, A AS A2, A AS A3, 
%     A AS A4, A AS A5
%WHERE A2.x = A1.z and 
%      A3.x = A1.y and A3.y = A1.z and
%      A4.x = A1.y and A4.z = A1.z and
%v     A5.x = A1.z and A5.y = A1.y and 
%      A5.z = A1.x;
%\end{verbatim}
The query translator generates the nested loop join for the input query. The order
of the loops are chosen such that the iteration space is as small as possible
and the conditions are levelled, i.e., hoisted to the outer most loop. For example,
the query planner of the query translator could choose following nested loop join 
for the input query,
\begin{center}
\begin{adjustbox}{bgcolor=green!10}
\begin{minipage}{0.8\textwidth}
\begin{tabbing}
\hspace{0.9cm}\=\hspace{0.1cm}\=\hspace{0.1cm}\=\hspace{0.1cm}\=\hspace{0.1cm}\=\hspace{0.1cm}\=\hspace{0.1cm}\=\hspace{0.1cm}\=\hspace{0.1cm}\=\hspace{0.1cm}\=\kill 
$\textit{loop}_1$: \> $\textbf{for all} \, t_1 \in A  \, \textbf{do}$ \\
$\textit{loop}_2$: \>\> $\textbf{for all} \, t_2 \in A : t_2(x) = t_1(z) \, \textbf{do}$ \\
$\textit{loop}_3$: \>\>\> $\textbf{for all} \, t_3 \in A : t_3(x) = t_1(y) \wedge t_3(y) = t_1(z)  \, \textbf{do}$ \\
$\textit{loop}_4$: \>\>\>\> $\textbf{for all} \, t_4 \in A : t_4(x) = t_1(y)\wedge t_4(z) = t_1(z) \, \textbf{do}$ \\
$\textit{loop}_5$: \>\>\>\>\> $\textbf{for all} \, t_5 \in A : t_5(x) = t_1(y)$ \\ 
$\                $\>\>\>\>\>$\wedge t_5(y) = t_1(y)\wedge t_5(z) = t_1(x) \, \textbf{do}$ \\
\>\>\>\>\>\>\textbf{if} $ t_1 \not \in B$ \textbf{ then} \\
\>\>\>\>\>\>\>\textbf{add } $t_1$ \textbf{ to } $B$\\
\>\>\>\>\>\>\textbf{endif}\\
\>\>\>\>\>\textbf{endfor}\\
\>\>\>$\ldots$\\
\>\textbf{endfor}
\end{tabbing}
\end{minipage}
\end{adjustbox}
\end{center}
where for each relation instance a loop in the nested loop join is generated. Where there is a variable that binds two attributes together, this is translated to loop 
predicate $t_2(x) = t_1(z)$. The outer most loop for this condition is the second outer most loop since it requires 
information of the tuple element $t_2(x)$ which only stabilises in the second outer-most loop. 
Other conditions of the example are placed accordingly. The projection function $\pi(t_1,\ldots, t_5)$ maps the tuples $t_1, \ldots, t_5$ to $t_1$, and the if-statement checks whether the tuple $t_1$ exists in the output relation $B$. If it does, it adds the tuple to $B$. 

The nested loop join is further transformed by the search rewriter to a $\searchquery$-nested loop join replacing loop predicates by primitive searches where possible. Since all loop predicates are conjunction of equality predicates, they can all be converted to search predicate of primitive searches. Hence, the output of the search rewriter is given below
\begin{center}
\begin{adjustbox}{bgcolor=green!10}
\begin{minipage}{0.4\textwidth}
\begin{tabbing}
\hspace{0.9cm}\=\hspace{0.1cm}\=\hspace{0.1cm}\=\hspace{0.1cm}\=\hspace{0.1cm}\=\hspace{0.1cm}\=\hspace{0.1cm}\=\hspace{0.1cm}\=\hspace{0.1cm}\=\hspace{0.1cm}\=\kill 
$\textit{loop}_1$: \> $\textbf{for all} \, t_1 \in A  \, \textbf{do}$ \\
$\textit{loop}_2$: \>\> $\textbf{for all} \, t_2 \in \searchquery_{x = t_1(y)}(A) \, \textbf{do}$ \\
$\textit{loop}_3$: \>\>\> $\textbf{for all} \, t_3 \in \searchquery_{x = t_1(y), y = t_1(z)}(A)  \, \textbf{do}$ \\
$\textit{loop}_4$: \>\>\>\> $\textbf{for all} \, t_4 \in \searchquery_{x = t_1(y), z = t_1(z)}(A) \, \textbf{do}$ \\
$\textit{loop}_5$: \>\>\>\>\> $\textbf{for all} \, t_5 \in \searchquery_{x = t_1(z), y = t_1(y), z = t_1(x)}(A) \, \textbf{do}$ \\
\>\>\>\>\>\>\textbf{if} $ t_1 \not \in B$ \textbf{ then} \\
\>\>\>\>\>\>\>\textbf{add } $t_1$ \textbf{ to } $B$\\
\>\>\>\>\>\>\textbf{endif}\\
\>\>\>\>\>\textbf{endfor}\\
\>\>\>$\ldots$\\
\>\textbf{endfor}
\end{tabbing}
\end{minipage}
\end{adjustbox}
\end{center}
for which the remaining loop predicate $C''_1,\ldots,C''_5$ hold independent of the 
tuples $t_1, \ldots, t_5$, and are omitted in the $\searchquery$-nested loop join above.

In the last step, the primitive searches are rewritten to range searches. The 
conversion from a $\searchquery$-nested loop join to a $\rangequery$-nested loop join
is performed by the index-optimiser. The index optimiser performs the 
conversion for each input relation separately. Since we only have  one 
input relation (i.e. relation $A$) for our input query, we have only one
invocation of the index optimizer. 
\end{example}

%\textbf{Datalog Engines.} Datalog and deductive databases have been pro-actively researched 
%in the database community~\cite{Ceri.GT.89,Ramakrishnan.SS.92,Ramamohanarao.Harland.94,Ramakrishnan.Ullman.95}. For 
%a comprehensive introduction to Dataog we recommend Abiteboul et al.~\cite{Abiteboul95}. Recently, 
%Datalog has regained considerable interest, driven by different applications 
%including data integration, networking, and program analysis. A survey that includes these developments was published recently by Green 
%et al.~\cite{Green.HLZ.13}.

\section{Experiments}
\label{sec:experiments}
In this section we evaluate our technique, which we referred to as \emph{Auto-Indexing}. We have implemented  
our approach in the \souffle Datalog-based analysis tool. We perform two sets of experiments where we compare 
auto-indexing with a na\"ive index selection where we select index orders based on syntactical delcarations of 
relations. We have ommited experiments with no indexing as it results in a time-out in all the experiments and thus 
confirms the need for the MIA for any reasonably sized program analysis benchmark.

The first set of experiments are conducted on data sets extracted from the DaCapo06~\cite{DaCapo} and Julia~\cite{Julia} 
program benchmarks\footnote{availible from https://bitbucket.org/yanniss/doop-benchmarks} that represent medium to large programs. For each dataset we perform various points-to analyses taken 
from the Doop~\cite{Doop} program analysis framework. The Doop analyses are  of varying difficulty ranging from the least 
computational heavy context insensitive analyses (ci, ci+, ci++) and the more computationaly difficult a context-sensitive 
analyses (o1sh, o2s2h, 3os3h).

The second set of experiments are performed on two very large data sets from industry (OpenJDK 7) and four industrial 
analyses. Again we compare auto-indexing with a na\"ive selection.  Again, we have ommited  experiments with 
no-indexing. Such large use cases are the primary target of our approach and the catalyst for the development of 
\souffle and are typically too large, or at the very least, among the most difficult problems for most Datalog-based 
tools. We demonstrate that our technique is paramount to \souffle being able to perform such analyses on giga-scale 
datasets in a practical amount of time. Moreover, we show that the our reported times perform on a par with the reported 
results in~\cite{OOPSLA15} on the same analysis and dataset. For compressions of \souffle to other tools we refer the reader 
to~\cite{CAV, CC}.

\subsection{Doop Analysis Benchmarks}
The results of the improvements of \souffle our technique (Auto)  compared to 
\souffle with na\"ive orderings are summarized in Fig.~\ref{fig:doop1} for run-time improvements and in 
\ref{fig:doop2} for memory consumption improvements. These experiments were conducted on a 2.6GHz Intel(R) 
Core i5-3320M with 8GB physical RAM.

\begin{figure*}
\centering
\begin{minipage}[b]{0.47\hsize}
\centering
  \footnotesize 
   \begin{tikzpicture}[x=1.5cm,y=1.5cm]
    \begin{axis}[
		    title=Auto-Index Scheme Run-Time Speedup vs Na\"ive Scheme,
		    xlabel=Doop Analyses,
		    xmin=0.5, xmax=6.5,
		    ylabel=Auto-Index Relative Speedup,
		    ymin=0.8, ymax=4.1,
                    axis y line*=left,
                    axis x line*=bottom,
                    xticklabels={ci, ci+, ci++,1osh, 2os2h, 3os3h},xtick={1,...,6},
                    x tick label style={rotate=90,anchor=east},
		    grid=major,
	            legend entries={DaCapo06, Julia},
                    legend style={at={(0.2,0.8)} , font=\tiny}],

	    ]
            \addlegendentry{DaCapo06}
	    \addplot+[
		    only marks,
		    mark=+,
		     mark size=3,
	    ] table {data/dacapo.dat};
            \addlegendentry{Julia}
	    \addplot+[
		    only marks,
		    mark=x,
		     mark size=3,
	    ] table {data/juliart.dat};
     \addplot[mark=none, dashed] coordinates {(0,4.0) (6.2,4.0)};
     \node[above] at (3.5, 3.7) {Na\"ive timeout};
    \end{axis}
    \end{tikzpicture}
    \caption{Relative Run-Time Speedup on Doop Benchmarks\label{fig:doop1} where timeout is set to 90 minutes.}
%\end{figure}
\end{minipage}
\begin{minipage}[b]{0.47\hsize}
\centering
%\begin{figure}
  \footnotesize 
   \begin{tikzpicture}[x=1.5cm,y=1.5cm]
    \begin{axis}[
		    title=Auto-Index Scheme Memory Improvement vs Na\"ive Scheme,
		    xlabel=Doop Analyses,
		    xmin=0.5, xmax=6.5,
		    ylabel=Auto-Index Relative Memory Improvement,
		    ymin=0.8, ymax=1.8,
                    axis y line*=left,
                    axis x line*=bottom,
                    xticklabels={ci, ci+, ci++,1osh, 2os2h, 3os3h},xtick={1,...,6},
                    x tick label style={rotate=90,anchor=east},
		    grid=major,
	            legend entries={DaCapo06, Julia},
                    legend style={legend pos=north west,font=\tiny}]
	    ]
            \addlegendentry{DaCapo06}
	    \addplot+[
		    only marks,
		    mark=asterisk,
		     mark size=3,
	    ] table {data/dacapomem.dat};
            \addlegendentry{Julia}
	    \addplot+[
		    only marks,
		    mark=x,
		     mark size=3,
	    ] table {data/juliamem.dat};

    \end{axis}
    \end{tikzpicture}
    \caption{Relative Memory Usage Improvement on Doop Benchmarks\label{fig:doop2} with timed out data points ommited.}
\end{minipage}
\end{figure*}

The results contain data points for each program in the respective benchmark dataset for a particular points-to analysis 
(x-axis) and the relative improvement of auto-indexing comapred to naive indexing (y-axis). 

The results exlusively demonstrate improvements in both run-time and memory usage. The results indicate that both run-times and 
memory reductions increase with more difficult datasets, with reasonable index reductions of approx. 15\%. For the Dacapo06 
dataset we see average memory and run-time improvements of approx. 15\% on simpler analyses to approx. 35\%, for the more difficult 
context sensitive analyses with exception to the run-times of 3os3h where larger DaCapo benchmarks result in significant run-time 
degredations due to the analyses reaching close to the full RAM capacity. The Julia benchmarks largely resulted in timeouts 
using the na\"ve configuration while auto-index was able to compute a majority of the benchmarks on the availible computing power.

\subsection{Large Scale Industrial Case Study : OpenJDK 7}
In the following set of experiments we present a very large scale industrial use case. We evaluate 
\souffle with our approach on several industrial-scale sized analyses and very large inputs from the
OpenJDK 7 codebase. These sets of experiments were conducted on a server node with 18 core 
Intel(R) Xeon(R) CPU E5-2699 v3 at 2.30GHz and with 378GB of physical RAM. 

In Table~\ref{tab:quantitative} a set of industrial analysis are outlined along with their 
number of rules and relations. The first analysis, ci-nocfg, refers to a context insensitive 
points-to analysis without prior call graph construction. This represents the simplest 
analysis. The next analysis, ci-pts, which represents a context insensitive points-to analysis 
with call graph construction. The second last analysis, cs-pts, is a context sensitive points-to 
that is typically regarded as too difficult for Datalog-based solvers for such large data sets. The 
last analysis, security, is a security analysis, similar to~\cite{OracleSec} that is performed to ensure the 
security of the Java JDK. 

In Table~\ref{tab:problemsizes}, we outline the number of code characteristics relevant to our 
analysis for both industrial code bases. We use two data sets (also known as EDBs) for performing the static program analysis on relational representations of the Java Development Kits\footnote{Java and JDK are registered trademarks of Oracle and/or its affiliates. Other names may be trademarks of their respective owners.} (JDK) versions 7 as well as the 
Java package source code\footnote{The Java package source code is a subset of JDK including all sub-packages in \texttt{java.*}.}. The large dataset is the OpenJDK 7 b147 library that has 1.4 million program variables, whereas the small data-set has 210 thousand program variables. The algorithmic complexity of the static program analyses used in this paper vary but have at least a cubic worst-case complexity in the number of program variables.  For the OpenJDK 7, and 
the Java package, the output relation sizes can have up to giga-tuples of data with several relations 
containing hundreds of attributes. 

\begin{table}
\centering
\footnotesize
\begin{tabular}{l|r|r|r|r|r}
Benchmark & \# rul & \# rel  \\
\hline
ci-nocfg & 34  & 48  \\
ci-pts       & 160 & 260 \\
cs-pts       & 174 & 292 \\
security     & 359 & 600 
\end{tabular}
\caption{Quantitative Statistics of Analyses}
\label{tab:quantitative}
\end{table}

\begin{table}
\centering
\footnotesize
\begin{tabular}{l|r|r}
Input & OpenJDK 7 & Java \\
\hline
Variables&1,440,875&210,076\\
Invocations&591,262&81,515\\
Methods&162,026&29,209\\
Objects&184,352&21,998\\
Classes&16,102&6,972\\
\end{tabular}
\caption{OpenJDK7 b147 and Java Package Size}
\label{tab:problemsizes}
\end{table}

\begin{table*}
\centering
\begin{minipage}[b]{0.99\hsize}
\footnotesize
\begin{center}
\begin{tabular}{l|r|r|r|r|r|r|r|r|r}
& \multicolumn{3}{c|}{Run-time (s)}
& \multicolumn{3}{c|}{Memory (GByte)}
& \multicolumn{3}{c}{Index Inserts (Mega)} 
\\
\multicolumn{1}{c|}{Analysis} 
& 
\multicolumn{1}{c|}{Auto} & 
\multicolumn{1}{c|}{Na\"ive} & 
\multicolumn{1}{c|}{Ratio}
& 
\multicolumn{1}{c|}{Auto} & 
\multicolumn{1}{c|}{Na\"ive} & 
\multicolumn{1}{c|}{Ratio}
& 
\multicolumn{1}{c|}{Auto} & 
\multicolumn{1}{c|}{Na\"ive} & 
\multicolumn{1}{c}{Ratio}
\\
\hline
ci-ncfg &2.34  &3.72  &1.59&1.50 &1.94 &1.29&27.07 &52.5 &1.94 \\
ci-pts  &14.87 &17.42 &1.17&2.40 &3.13 &1.30&31.73  &51.05 &1.61 \\
cs-pts  &152.22&444.23&2.92&23.56&49.45&2.10&512.97  &971.10 &1.89 \\
security&17.31 &19.48 &1.13&2.77 &3.47 &1.25&45.01  &65.05 &1.45 \\
\hline
\end{tabular}
\end{center}
\caption{Java Package of OpenJDK7 b147}
\label{tab:small}
\end{minipage} \\
\centering
\begin{minipage}[b]{0.99\hsize}
\centering
\footnotesize
\begin{center}
\begin{tabular}{l|r|r|r|r|r|r|r|r|r}
& \multicolumn{3}{c|}{Run-time (s)}
& \multicolumn{3}{c|}{Memory (GByte)}
& \multicolumn{3}{c}{Index Inserts (Mega)}
\\
\multicolumn{1}{c|}{Benchmark} 
& 
\multicolumn{1}{c|}{Auto} & 
\multicolumn{1}{c|}{Na\"ive} & 
\multicolumn{1}{c|}{Ratio}
& 
\multicolumn{1}{c|}{Auto} & 
\multicolumn{1}{c|}{Na\"ive} & 
\multicolumn{1}{c|}{Ratio}
& 
\multicolumn{1}{c|}{Auto} & 
\multicolumn{1}{c|}{Na\"ive} & 
\multicolumn{1}{c}{Ratio} 
\\
\hline
ci-ncfg&64.32&140.30&2.18&31.14&43.55&1.40 &1447.65 &2882.13 &1.99 \\
ci-pts&353.55&656.04&1.86&26.57&42.97&1.62 &1020.39 &1689.23 &1.66 \\
cs-pts&24248.60&n/a&n/a&825.77&n/a&n/a     &19578.90 &39874.23 &2.04 \\
security&52025.00&n/a&n/a&75.30&n/a&n/a    &15.33 &18.44 &1.2 \\\hline
\end{tabular}
\end{center}
\caption{\footnotesize OpenJDK 7 b147 dataset}
\label{tab:large}
\end{minipage}
\end{table*}

The benchmark results for the small data and large data sets are shown in Table~\ref{tab:small} 
and Table~\ref{tab:large}, respectively. 

For the large benchmark, the programs 
are executed in parallel using substantial amounts of memory and run-time, i.e., 826GB and 14.5 hours. Note 
that the  na\"ive approach was not computable for the benchmark ``cs-pts'' and ``security'' because of lack of 
memory and/or running into computation limits. The work of the na\"ive approach is between 1.2 and 2.04 times 
more than our new technique that results in speedups between 1.86 and 2.18 depending on the benchmark. For the 
small dataset all static program analyses are computable. The speedup of our new technique is 
up to 2.92 and up to 2.1 times less memory is used. 

The JDK experiments resulted in a maximum 
memory reduction of 2.1 times less memory and a maximum 
speedup of 2.92. As noted before, the OpenJDK 7 data set resulted 
in a timeout for the context-insensitive and security analysis with 
the na\"ive technique while our technique enabled the JDK library to 
be processed. The memory improvement is due to minimizing redundant 
index data structures. We attribute the run-time performance to index
maintenance costs. Moreover, we highlight that the times achieved using 
our approach are on a par with the results in~\cite{OOPSLA15} using the same 
data set and analysis. On the java-points-to 
analysis used for the DaCapo benchmarks our approach managed to provide speed-ups and 
memory reductions of approx. 20\%.

Overall, the experiments show that our technique of automatically 
generating minimal indices significantly improves both memory 
usage and run-time in the resultant analyzer. It is no surprise that 
our technique is more effective on larger, more complex analyses such as the 
ones in the JDK benchmarks, as such analyses provide more opportunities for 
index reductions. Such analyses are in fact the main motivation of our technique, 
in other words, our main motivation is to enable 
Datalog-based analysis on large industrial sized
analysis and data sets, that are typically deemed too difficult for other Datalog-based 
tools. On the other hand, our approach on the relatively small java-points-to 
analysis still managed to provide significant speed-ups and 
memory reductions. While not on the scale of the large scale analysis improvements, 
we believe this result still has merit and is a notable improvement nevertheless.

\section{Related Work}
\label{sec:relatedwork}

\subsection{Datalog Engines}
Datalog has been pro-actively researched 
in several computer science 
communities~\cite{Ceri.GT.89,Ramakrishnan.SS.92,Ramamohanarao.Harland.94,Ramakrishnan.Ullman.95}. For 
a comprehensive introduction to Dataog we recommend~\cite{Abiteboul95}. Recently, 
Datalog has regained considerable interest, driven by different applications 
including data integration, networking, and program analysis. A survey that includes these developments 
was published recently by~\cite{Green.HLZ.13}. There is likewise a large body of work on 
Datalog evaluation and Datalog compilation. We refer the reader to the related work of~\cite{CC} for a 
general overview of bottom-up engines. Here, we mention some notable 
state-of-the-art Datalog engines with ordered data structures, namely, LogicBlox~\cite{Green.AK.12-logicblox} and 
bddbddb~\cite{Whaley05a}. Unlike bddbddb we do not create a global variable order but a per relation order based 
on a nested loop join, thus solving a polynomial-time problem. 

We also believe our MOSP solution can be beneficial to other Datalog engines, even query engines that 
interpret the Datalog rule can perform the MOSP computation online, since the overhead for the computation
is very small we believe this should not greatly affect performance while potentially improving the 
evaluation significantly. Moreover, we believe that our approach could be in special cases 
applicable to general query engines that may not have Datalog as a front-end language. For example, our 
approach could work for bottom-up engines that use SQL defined queries. 

Comparisons to LogicBlox: As LogicBlox has traditionally been the default engine of Doop, we can foresee 
potential interest in comparisons between LogicBlox and \souffle on Doop benchmarks. In this paper we have 
omitted such a comparison due to a recent in-depth paper by the authors if Doop~\cite{doopsouffle}. This article  
gives an extensive overview of the differences between the two engines and provides an indepth
comparisson of \souffle (with auto-index turned on by default) and LogicBlox. The article demonstrates 
significant improved run-times, for which believe our technqiue plays a significant factor. Our 
own experiments did not demonstrate any significant deviation from their results in \cite{doopsouffle}. We only 
point out that in our experiments \souffle had lower memory usages compared to LogicBlox.

\subsection{Join Algorithms}
Join computations are vital for fast Datalog execution. As index selection and join order scheduling 
are intertwined for query optimization, many state-of-the-art Datalog engines have separated the two 
for practical reasons. For example, engines such as Socialite~\cite{Socialite} and \souffle~\cite{CAV} allow users 
to specify the join order. For such engines our technique can be used. Other Datalog engines such 
as Logicblox~\cite{logicblox} use a leapfrog join that, while alleviating users from specifying join order, 
require users to specify indices manually. Our technique rests on the observation that large Datalog programs 
usually comprises of significantly less rules than relations. Moreover, of all these rules,
usually only a few require manual loop scheduling. These can be identified using a profiler\footnote{\souffle profiler for example}, 
or alternatively, loop schedules can be automated using heuristic techniques~\cite{Selingers}. Therefore, our preference is to 
fix loop orders rather then indices for a better user experience. The preference for this design choice has been 
confirmed from the \souffle user community.

\subsection{Index Selection}
In the context of traditional relational databases, the general problem of automatically selecting indices for a set of 
queries, referred to in the literature as the index selection problem (ISP)~\cite{Schkolnick1975,Comer1978,Ip, Genetic}, is 
well studied and has been shown to be NP-hard~\cite{Piatetsky-Shapiro83}. State-of-the-art 
techniques typically formulate ISP as a variation of the $0$-$1$ knapsack problem. Knapsack formulations 
of ISP assume a relational database setup, i.e., relations are stored in secondary storage and queries 
are executed on normalized data. Indices are ``packed'' by balancing the overall execution time of queries for 
an index configuration (i.e., a subset of indices that influence the performance of a query), and the costs of index 
maintenance.  All possible index configurations are reflected in form of $0-1$ decision variables. Constraints in 
the index selection problem ensure that at most one index configuration is chosen for a query, and other constraint 
forces an index to be packed if the index occurs in a selected index configuration of a query. The objective function 
of ISP is the profit of selected index configuration of queries minus the cost of maintaining selected indices. Note 
that if the relations are normalized, the set of indexes will be small in relational databases, and hence ISP as a 
combinatorial optimisation problem becomes viable.

The standard ILP model for 
ISP may be expressed as the following ILP model below:
\begin{align*}
\max &= \sum_{q \in Q} \sum_{k \in K} g_{qk} x_{qk}- \sum_{i \in I} f_i y_i \hspace*{-3cm} \\
\textit{s.t.} & \sum_{k \in K} x_{qk} \leq 1, & \forall  k \in K\\
&x_{qk} \leq y_i,  & \forall  q \in Q, k \in K, i \in I: i \in k \\
&x_{qk}, y_i \in \{0,1\}, & \forall q \in Q, k \in K, i \in I \\
\end{align*}
Here $I$ be the set of indices, $Q$ the
set of all queries, and $K$ are the index configurations. An index configuration
$k$ is a subset of the index set $I$. 
The $0$-$1$ decision variable $x_{qk}$  determines whether configuration $k$ for query  $q$ is selected, and the 
$0$-$1$ variable $y_i$ determines whether index $i$ is selected. The first constraints of the ILP model ensures 
that at most one index configuration is chosen for a query $q$, and the second constraint forces variable $y_i$ to 
be set to one if there exists a decision variable $x_{qk}$ whose configuration $k$ contains index $i$. 
The objective function is the profit of the configurations over all queries minus the cost of maintaining indices 
where $g_{qk}$ is the profit of index configuration $k$ for  query $q$,  and $f_i$ are the maintenance costs for 
index $i$. 

Automated index selection was 
studied in the context of self-tuning 
database systems \cite{Chaudhuri07}. As choosing an optimal index
configuration is NP-complete in the general case \cite{Piatetsky-Shapiro83},
the mechanisms that were introduced starting in the late 90s used the
existing cost model and some heuristics to select/propose indices for
a workload. Index advisors were implemented in all major relational
database systems \cite{Chaudhuri97,Valentin00,Dageville04}. Other models uses the 
Knapsack problem as a vehicle for finding an optimal solution~\cite{Ip} and more exotic approaches use 
genetic algorithms~\cite{Genetic} to find sub-optimal indices. The work of \cite{Gundem1999111} proposes an 
extension of prior index optimisation models where there are multiple-candidates available
for attributes, and G\"{u}ndem shows that the optimisation model is 
NP-hard, and provide an approximation algorithm which is bounded
by a logarithmic time order. Older works~\cite{Schkolnick1975,Comer1978} provide
some probabilistic modelling for selecting secondary indices and some 
ad-hoc approaches. A related problem is the index selection for views for example in 
OLAP domains~\cite{Gupta97}. 
Our index selection problem differs from the classic ISP literate and to the best of our knowledge is the 
first formulation of such as problem. In our case, we are restricted to support primitive searches only. Primitive 
searches occur in equi-joins and simple value queries. We further have a maximal index assumption, i.e., each primitive search is
covered by at least one index. This assumptions is important for high-performance
systems which needs to accelerate all searches and the question is to minimise
the indices to cover all searches.  Since our problem differs to existing approaches, we obtain an optimal 
algorithm that exhibits a polynomial worst-case execution time. The effect of our 
optimisation reduces the cost of maintaining indices. It relies on a pre-determined clause body order and assumes all relations are 
indexed in memory. With modern memory systems this assumption becomes very much feasible and leads to large performance 
improvements. Secondly, the nature of Datalog restricts search predicate of primitive searches, i.e., the 
search predicate has to be an equality predicate over the attributes 
of the relation. The traditional formulations of the model do not capture this restrictions as they are too 
coarse-grained in a mathematical sense. While not related to index selection the approach in~\cite{Jag90} 
uses Dilworths theorem as a fast reachibility algorithm to process graph reachability queries.

%\section{Related Work}
%\label{sec:relatedwork}

\section{Conclusion}
\label{sec:conclusion}
We have presented an join computation technique that computes the minimal set of indices a given relation. The 
proposed algorithm runs in polynomial time due to exploiting a relationship
to Dilworth's problem: Instead of finding indices directly, and searching in a double-exponential space for finding them, we construct a partial order over the
primitive searches. A minimum chain cover in the partial order over the primitive 
searches construct the minimal set of indices, indirectly.
We have 
demonstrated the feasibility of our approach through experiments using \souffle, a Datalog-based 
static analysis tool. In future, we plan to explore the possibility of adapting our technique to general in-memory relational database engines, as well. 

% ensure same length columns on last page (might need two sub-sequent latex runs)
\balance

%ACKNOWLEDGMENTS are optional
\section{Acknowledgments}
We would like to thank Oracle Labs, Alan Fekete, Byron Cook, Chenyi Zhang and all our anonymous 
reviewers.

% The following two commands are all you need in the
% initial runs of your .tex file to
% produce the bibliography for the citations in your paper.
\bibliographystyle{abbrv}
\bibliography{main}

\begin{thebibliography}{10}

\bibitem{Abiteboul95}
S.~Abiteboul, R.~Hull, and V.~Vianu.
\newblock {\em Foundations of Databases}.
\newblock Addison-Wesley, 1995.

\bibitem{doopsouffle}
T.~Antoniadis, K.~Triantafyllou, and Y.~Smaragdakis.
\newblock Porting doop to souffl\&\#xe9;: A tale of inter-engine portability
  for datalog-based analyses.
\newblock In {\em Proceedings of the 6th ACM SIGPLAN International Workshop on
  State Of the Art in Program Analysis}, SOAP 2017, pages 25--30, New York, NY,
  USA, 2017. ACM.

\bibitem{Bierhoff05}
K.~Bierhoff.
\newblock Alias analysis with bddbddb, star project report, 17-754 analysis of
  software artifacts.
\newblock
  \url{www.cs.cmu.edu/~aldrich/courses/654/tools/bierhoff-bddbddb-05.pdf},
  2005.

\bibitem{DaCapo}
S.~M. Blackburn, R.~Garner, C.~Hoffman, A.~M. Khan, K.~S. McKinley, R.~Bentzur,
  A.~Diwan, D.~Feinberg, D.~Frampton, S.~Z. Guyer, M.~Hirzel, A.~Hosking,
  M.~Jump, H.~Lee, J.~E.~B. Moss, A.~Phansalkar, D.~Stefanovi\'{c},
  T.~{VanDrunen}, D.~von Dincklage, and B.~Wiedermann.
\newblock The {DaCapo} benchmarks: {J}ava benchmarking development and
  analysis.
\newblock In {\em OOPSLA '06: Proceedings of the 21st annual ACM SIGPLAN
  conference on Object-Oriented Programing, Systems, Languages, and
  Applications}, pages 169--190, New York, NY, USA, Oct. 2006. ACM Press.

\bibitem{SKP15}
B.Scholz, K.Vorobyov, P.Krishnan, and T.Westmann.
\newblock A datalog source-to-source translator for static program analysis: An
  experience report.
\newblock {\em Australasian Software Engineering Conference}, 2015.

\bibitem{Ceri89}
S.~Ceri, G.~Gottlob, and L.~Tanca.
\newblock What you always wanted to know about datalog (and never dared to
  ask).
\newblock {\em IEEE Trans. on Knowl. and Data Eng.}, 1(1):146--166, Mar. 1989.

\bibitem{Ceri.GT.89}
S.~Ceri, G.~Gottlob, and L.~Tanca.
\newblock What you always wanted to know about datalog (and never dared to
  ask).
\newblock {\em IEEE Transactions on Knowledge and Data Engineering},
  1(1):146--166, 1989.

\bibitem{Chau97}
S.~Chaudhuri and V.~Narasayya.
\newblock An efficient, cost-driven index selection tool for microsoft sql
  server.
\newblock Very Large Data Bases Endowment Inc., August 1997.

\bibitem{Chaudhuri07}
S.~Chaudhuri and V.~Narasayya.
\newblock Self-tuning database systems: A decade of progress.
\newblock In {\em Proceedings of the 33rd International Conference on Very
  Large Data Bases}, VLDB '07, pages 3--14. VLDB Endowment, 2007.

\bibitem{Chaudhuri97}
S.~Chaudhuri and V.~R. Narasayya.
\newblock An efficient cost-driven index selection tool for microsoft sql
  server.
\newblock pages 146--155, Athens, Greece, 1997. Morgan Kaufmann.

\bibitem{OracleSec}
C.~Cifuentes, A.~Gross, and N.~Keynes.
\newblock Understanding caller-sensitive method vulnerabilities: a class of
  access control vulnerabilities in the java platform.
\newblock In A.~M{\o}ller and M.~Naik, editors, {\em Proceedings of the 4th
  {ACM} {SIGPLAN} International Workshop on State Of the Art in Program
  Analysis, SOAP@PLDI 2015, Portland, OR, USA, June 15 - 17, 2015}, pages
  7--12. {ACM}, 2015.

\bibitem{Codd70}
E.~F. Codd.
\newblock A relational model of data for large shared data banks.
\newblock {\em Commun. ACM}, 13(6):377--387, June 1970.

\bibitem{Comer1978}
D.~Comer.
\newblock The difficulty of optimum index selection.
\newblock {\em ACM Trans. Database Syst.}, 3(4):440--445, Dec. 1978.

\bibitem{Dageville04}
B.~Dageville, D.~Das, K.~Dias, K.~Yagoub, M.~Za{\"{\i}}t, and M.~Ziauddin.
\newblock Automatic {SQL} tuning in oracle 10g.
\newblock pages 1098--1109, Toronto, Canada, September 2004. Morgan Kaufmann.

\bibitem{OOPSLA15}
J.~Dietrich, N.~Hollingum, and B.~Scholz.
\newblock Giga-scale exhaustive points-to analysis for java in under a minute.
\newblock In J.~Aldrich and P.~Eugster, editors, {\em Proceedings of the 2015
  {ACM} {SIGPLAN} International Conference on Object-Oriented Programming,
  Systems, Languages, and Applications, {OOPSLA} 2015, part of {SPLASH} 2015,
  Pittsburgh, PA, USA, October 25-30, 2015}, pages 535--551. {ACM}, 2015.

\bibitem{Dilworth50}
R.~{Dilworth}.
\newblock {A decomposition theorem for partially ordered sets.}
\newblock {\em {Ann. Math. (2)}}, 51:161--166, 1950.

\bibitem{Fulkerson56}
D.~R. Fulkerson.
\newblock Note on dilworth's decomposition theorem for partially ordered sets.
\newblock {\em Proceedings of the American Mathematical Society}, 7(4):pp.
  701--702, 1956.

\bibitem{Green.AK.12-logicblox}
T.~J. Green, M.~Aref, and G.~Karvounarakis.
\newblock {LogicBlox}, platform and language: A tutorial.
\newblock In {\em Datalog in Academia and Industry - Second International
  Workshop}, volume 7494 of {\em Lecture Notes in Computer Science}, pages
  1--8. Springer, September 2012.

\bibitem{Green.HLZ.13}
T.~J. Green, S.~S. Huang, B.~T. Loo, and W.~Zhou.
\newblock Datalog and recursive query processing.
\newblock {\em Foundations and Trends in Databases}, 5(2):105--195, 2013.

\bibitem{Groff09}
J.~Groff and P.~Weinberg.
\newblock {\em SQL The Complete Reference, 3rd Edition}.
\newblock McGraw-Hill, Inc., New York, NY, USA, 3 edition, 2010.

\bibitem{Gupta97}
H.~Gupta, V.~Harinarayan, A.~Rajaraman, and J.~D. Ullman.
\newblock Index selection for olap.
\newblock pages 208--219, 1997.

\bibitem{Gundem1999111}
T.~Gündem.
\newblock Near optimal multiple choice index selection for relational
  databases.
\newblock {\em Computers and Mathematics with Applications}, 37(2):111 -- 120,
  1999.

\bibitem{Hoder11}
K.~Hoder, N.~Bj{\o}rner, and L.~De~Moura.
\newblock {\it Z}: An efficient engine for fixed points with constraints.
\newblock In {\em Proceedings of the 23rd international conference on Computer
  aided verification}, CAV'11, pages 457--462, Berlin, Heidelberg, 2011.
  Springer-Verlag.

\bibitem{Ip}
M.~Ip, L.~Saxton, and V.~Raghavan.
\newblock On the selection of an optimal set of indexes.
\newblock {\em Software Engineering, IEEE Transactions on}, SE-9(2):135--143,
  March 1983.

\bibitem{Jag90}
H.~V. Jagadish.
\newblock A compression technique to materialize transitive closure.
\newblock {\em ACM Trans. Database Syst.}, 15(4):558--598, Dec. 1990.

\bibitem{CAV}
H.~Jordan, B.~Scholz, and P.~Subotic.
\newblock Souffl{\'{e}}: On synthesis of program analyzers.
\newblock In S.~Chaudhuri and A.~Farzan, editors, {\em Computer Aided
  Verification - 28th International Conference, {CAV} 2016, Toronto, ON,
  Canada, July 17-23, 2016, Proceedings, Part {II}}, volume 9780 of {\em
  Lecture Notes in Computer Science}, pages 422--430. Springer, 2016.

\bibitem{Genetic}
J.~Kratica, I.~Ljubic, and D.~To\v{s}ic.
\newblock A genetic algorithm for the index selection problem.
\newblock In {\em Proceedings of the 2003 International Conference on
  Applications of Evolutionary Computing}, EvoWorkshops'03, pages 280--290,
  Berlin, Heidelberg, 2003. Springer-Verlag.

\bibitem{logicblox}
I.~LogicBlox.
\newblock Declartive cloud platform for applications that combine transactions
  \& analytics.
\newblock \url{http://www.logicblox.com}.

\bibitem{GlueNail}
G.~Phipps, M.~A. Derr, and K.~A. Ross.
\newblock Glue-nail: A deductive database system.
\newblock In {\em Proceedings of the 1991 ACM SIGMOD International Conference
  on Management of Data}, SIGMOD '91, pages 308--317, New York, NY, USA, 1991.
  ACM.

\bibitem{Piatetsky-Shapiro83}
G.~Piatetsky-Shapiro.
\newblock {The Optimal Selection of Secondary Indices is NP-complete}.
\newblock {\em SIGMOD Rec.}, 13(2):72--75, Jan. 1983.

\bibitem{Pijls2013}
W.~Pijls and R.~Potharst.
\newblock Another note on dilworth's decomposition theorem.
\newblock {\em Journal of Discrete Mathematics}, 2013:4, 2013.

\bibitem{Ramakrishnan.SS.92}
R.~Ramakrishnan, D.~Srivastava, and S.~Sudarshan.
\newblock Efficient bottom-up evaluation of logic programs.
\newblock In P.~Dewilde and J.~Vandewalle, editors, {\em Computer Systems and
  Software Engineering}, pages 287--324. Springer US, 1992.

\bibitem{Coral}
R.~Ramakrishnan, D.~Srivastava, S.~Sudarshan, and P.~Seshadri.
\newblock Implementation of the coral deductive database system.
\newblock {\em SIGMOD Rec.}, 22(2):167--176, June 1993.

\bibitem{Ramakrishnan.Ullman.95}
R.~Ramakrishnan and J.~D. Ullman.
\newblock A survey of deductive database systems.
\newblock {\em Journal of Logic Programming}, 23(2):125--149, 1995.

\bibitem{Ramamohanarao.Harland.94}
K.~Ramamohanarao and J.~Harland.
\newblock An introduction to deductive database languages and systems.
\newblock {\em Journal of Very Large Data Bases}, 3(2):107--122, 1994.

\bibitem{Schkolnick1975}
M.~Schkolnick.
\newblock The optimal selection of secondary indices for files.
\newblock {\em Information Systems}, 1(4):141 -- 146, 1975.

\bibitem{CC}
B.~Scholz, H.~Jordan, P.~Subotic, and T.~Westmann.
\newblock On fast large-scale program analysis in datalog.
\newblock In A.~Zaks and M.~V. Hermenegildo, editors, {\em Proceedings of the
  25th International Conference on Compiler Construction, {CC} 2016, Barcelona,
  Spain, March 12-18, 2016}, pages 196--206. {ACM}, 2016.

\bibitem{Selingers}
P.~G. Selinger, M.~M. Astrahan, D.~D. Chamberlin, R.~A. Lorie, and T.~G. Price.
\newblock Access path selection in a relational database management system.
\newblock In {\em Proceedings of the 1979 ACM SIGMOD International Conference
  on Management of Data}, SIGMOD '79, pages 23--34, New York, NY, USA, 1979.
  ACM.

\bibitem{Socialite}
J.~Seo, J.~Park, J.~Shin, and M.~S. Lam.
\newblock Distributed socialite: A datalog-based language for large-scale graph
  analysis.
\newblock {\em Proc. VLDB Endow.}, 6(14):1906--1917, Sept. 2013.

\bibitem{Doop}
Y.~Smaragdaiks, M.~Bravenboer, and G.~Kastrinis.
\newblock Doop: A framework for java pointer analysis.
\newblock \url{http://doop.program-analysis.org/}.

\bibitem{Julia}
F.~Spoto.
\newblock {\em The Julia Static Analyzer for Java}, pages 39--57.
\newblock Springer Berlin Heidelberg, Berlin, Heidelberg, 2016.

\bibitem{Valentin00}
G.~Valentin, M.~Zuliani, D.~C. Zilio, G.~M. Lohman, and A.~Skelley.
\newblock {DB2} advisor: An optimizer smart enough to recommend its own
  indexes.
\newblock In {\em {ICDE}}, pages 101--110, 2000.

\bibitem{Whaley05a}
J.~Whaley, D.~Avots, M.~Carbin, and M.~S. Lam.
\newblock Using datalog with binary decision diagrams for program analysis.
\newblock In {\em Proceedings of the Third Asian conference on Programming
  Languages and Systems}, APLAS'05, pages 97--118, Berlin, Heidelberg, 2005.
  Springer-Verlag.

\end{thebibliography}
% You must have a proper ".bib" file
%  and remember to run:
% latex bibtex latex latex
% to resolve all references
%APPENDIX is optional.
% ****************** APPENDIX **************************************
% Example of an appendix; typically would start on a new page
%pagebreak

\begin{appendix}
\section{Proofs}
\subsection{Proofs from Section~\ref{sec:cover}}
\begin{proof}[Lemma~\ref{lemma:rangequerycover}]
Let us assume a index of size $k$ i.e.,  $\ell = x_1 \prec \dots \prec x_k$ and a relation $R$ of a fixed size.
By induction on size of $\ell$ we have:
Base case (size $1$) : it is trivial that $$ \{ t \in R \ | \ t(x_1) = v_1 \} = \{ t \in R \ | \ lb(v_1) \sqsubseteq_{x_1} t \sqsubseteq_{x_1} ub(v_1) \}.$$
Now we construct the IH (size is $n$) we hence 
\begin{align*}
&\{ \ t \in R \ | \ t(x_1) = v_1 \wedge \dots \wedge t(x_n) = v_n \} = \\
&\{ \ t \in R \ | \ lb(v_1, \dots, v_n) \sqsubseteq_{\ell} t \sqsubseteq_{\ell} ub(v_1, \dots, v_n) \} \text{ where } \\
&\ell = x_1 \prec \dots \prec x_n
\end{align*}

Assuming the IH, and letting $k = n + 1$ we have

\begin{align*}
&\{ \ t \in R \ | \ t(x_1) = v_1 \wedge \dots \wedge t(x_n) = v_n \wedge t(x_{n+1}) \} = \\
&\{ \ t \in R \ | \ lb(v_1, \dots, v_n, v_{n+1}) \sqsubseteq_{\ell} t \sqsubseteq_{\ell} ub(v_1, \dots, v_n, v_{n+1}) \}.
\end{align*}

We can rewrite this as: 

\begin{align*}
& \{ \ t \in \sigma_{x_{n+1}}(R) \ | \ t(x_1) = v_1 \wedge \dots \wedge t(x_n) = v_n \} = \\
& \{ \ t \in \sigma_{\rho(x_{n+1}, lb(v_{n+1}), ub(v_{n+1}))}(R)  \ | \\
& \ lb(v_1, \dots, v_n) \sqsubseteq_{x_{n+1}} t \sqsubseteq_{x_{n+1}} ub(v_1, \dots, v_n) \}.  
\end{align*}
Hence by IH and the base case it holds.\\

If $\ell$ is $x_1 \prec \dots \prec x_k \prec x_s$ the property still holds as the bounds for the $x_s$ are set to the 
infima ($\bot$) and suprema ($\top$) for $a$ and $b$, respectively. Hence retaining the property.
\end{proof}

\subsection{Proofs from Section~\ref{sec:mosp}}

\subsubsection{Lemma Proofs}

\begin{proof}[of Lemma~\ref{lem:cardinality}]
We can bound the cardinality of the set of all possible sequences
as follows, 
\begin{align*}
m! = |\perm(A)| & \leq \left | \bigcup_{X \subseteq A, X\not=\emptyset} \perm(X) \right | \\
   &=  \sum_{X \subseteq A, X\not=\emptyset} |X|! \\ 
  &= \sum_{1\leq i \leq m} \binom{m}{i} i!   \\
  &= m! \sum_{1\leq i \leq m} \frac{1}{(m-i)!} \\
  &= m! \sum_{0 \leq i \leq m-1} \frac{1}{i!}  \\
  &\leq m! \sum_{ i \geq 0 } \frac{1}{i!} \\
  & =  \euler \cdot m! 
\end{align*}
The lower bound is given by the $m!=|\perm(X)|$ since $\perm(X)\subseteq
\lexset$.  For the upper bound, we sum up the cardinalities of
permutations over all subsets except the empty set.  We reorder the sum
by summing up the permutations of subsets that have cardinality $i$, i.e., there are $\binom{\tl}{i}$ subsets of cardinality $i$. By simplifying the binom and factoring out
$m!$ we can rearrange the summation such that the index runs from $0$
to $\tl-1$.  Since the numbers of the series $\frac{1}{i!}$ are
positive and converge, we can extend the range to infinity to obtain
an upper bound and the sum converges to the Euler number. Hence,
$\euler \cdot \tl!$ is an upper bound on the number of sequences
for $m$ attributes.
\end{proof}

\begin{proof}[of Lemma~\ref{lemma:chaincard}]
By construction
\end{proof}

\begin{proof}[of Lemma~\ref{lemma:set-size}]
By contradiction. \emph{(Case 1).} Assume $|\alpha_{\searchset}^{1}(L)|>|L|.$
This is not possible because function $\alpha_{\searchset}^{0}$ maps one index
to exactly one chain. Hence, due to the pigeonhole principle there
can be at most $|L|$ different chains in $\alpha_{\searchset}^{1}(L)$. \emph{(Case
2).} Assume $|\alpha_{\searchset}^{1}(L)|<|L|.$ Hence there must exist $l_{1}\not=l_{2}\in L$,
such that $\alpha_{\searchset}^{0}(l_{1})=\alpha_{\searchset}^{0}(l_{2})$, and $\alpha_{\searchset}^{1}(L\setminus\{l_{2}\})=\alpha_{\searchset}^{1}(L)$.
However, due to Lemma~\ref{lemma:gamma-cover}, $\textbf{c-cover}(\alpha_{\searchset}^{1}(L))\Rightarrow\textbf{c-cover}_{\searchset}(\alpha_{\searchset}^{1}(L\setminus\{l_{2}\})\Rightarrow\textbf{l-cover}_{\searchset}(L\setminus\{l_{2}\})$.
Hence, $L\not\in f(\searchset)$ because $|L|$ is not minimal, i.e., $L\setminus\{l_{2}\}$
represents a smaller solution for which $\textbf{l-cover}(L\setminus\{l_{2}\})$
holds. As a consequence of \emph{(Case 1)} and \emph{(Case 2),} the
lemma holds. 
\end{proof}

\begin{proof}[of Lemma~\ref{lemma:cover}]
By definition. 
\begin{eqnarray*}
\textbf{l-cover}_{\searchset}(L) & \Leftrightarrow & \forall s\in \searchset:\exists l\in L:\textbf{prefix}_{|s|}(l)=s\\
 & \Leftrightarrow & \forall s\in \searchset:\exists l\in L:s\in\alpha_{\searchset}^{0}(l)\\
 & \Leftrightarrow & \forall s\in \searchset:\exists c\in\alpha_{\searchset}^{1}(L):s\in c\\
 & \Leftrightarrow & \textbf{c-cover}(\alpha_{\searchset}^{1}(L))
\end{eqnarray*}
\end{proof}

\begin{proof}[of Lemma~\ref{lemma:gamma-cover}]
By definition. 
\begin{eqnarray*}
 & &\textbf{c-cover}_{\searchset}(C)   \Rightarrow \forall \search\in \searchset:\exists c\in C:\search\in c\\
 & \Rightarrow&  \forall L \in \gamma_{\searchset}^{1}(C) : \forall \search\in \searchset:\exists c\in C:\search\in c \\
 & \Rightarrow&  \forall L \in \{ \{\ell_1, \dots, \ell_k \} \ \mid \ \ell_1 \in \gamma_{\searchset}^0(c_1), \dots, \ell_k \in \gamma_{\searchset}^0(c_k) \} : \\ 
 & & \forall s \in Q : \exists c \in C : s\in c\\
 & \Rightarrow&  \forall L \in \{ \{\ell_1, \dots, \ell_k \} \ \mid \ c_1 \subseteq \alpha_{\searchset}^0(\ell_1), \dots, c_k \subseteq \alpha_{\searchset}^0(\ell_k) \} : \\ & & \forall s \in \searchset : \exists c \in C : s\in c\\
 & \Rightarrow&   \forall L \in \gamma_{\searchset}^{1}(C) : \forall s\in \searchset:\exists l\in L:s\in \alpha_{\searchset}^{0}(\ell) \\
 & \Rightarrow&  \forall L\in\gamma_{\searchset}^{1}(C):\textbf{l-cover}_{\searchset}(L)
\end{eqnarray*}
\end{proof}

\subsubsection{Theorem Proofs}

\begin{proof}[of Theorem~\ref{thm:alpha-optimal}]
%By contradiction. Assume $\exists L\in f_{\searchset}:\alpha_{\searchset}^{1}(L)\not\in g_{\searchset}$.
%We introduce $C\equiv\alpha_{\searchset}^{1}(L)$. By previous Lemma~\ref{lemma:set-size}
%and Lemma~\ref{lemma:cover}, $|C|=|L|$ and $\textbf{c-cover}_{\searchset}(C)$
%hold, i.e., $C$ is a chain cover of size $|L|$ but cannot be optimal.
%Hence, $\exists C'\in g_{\searchset}:|C'|<|C|$ and $\exists L'\in\gamma_{\searchset}^{1}(C')$
%and $|L'|=|C'|$ that is a cover violating the optimally assumption ($L \in f_{\searchset}$)
%of $L$.
If $L \in f_{\searchset}$ then $L$ is a l-cover and is minimal. By Lemma~\ref{lemma:set-size}
and Lemma~\ref{lemma:cover}, cardinality and cover is preserved. Hence we have a chain 
$C\equiv\alpha_{\searchset}^{1}(L)$ that is a c-cover and minimal due to the fact that if it was not 
minimal, there would be an $|L'| < |L|$ by Lemma~\ref{lemma:set-size}. However, this would violate 
our assumption that L is in $f_{\searchset}$.
\end{proof}

\begin{proof}[of Theorem~\ref{thm:main2}]
We instantiate $C_i$ and $L_i$ from the domains of $g_{\searchset}$ and $\gamma_{\searchset}^1(C_i)$, 
respectively. Given that \textbf{c-cover}$_{\searchset}$($C_i$) holds and that $\forall C \in g_{\searchset} / C_i : |C| = |C_i|$, 
we want to show \textbf{l-cover}($L_i$) holds and that $\forall L \in \gamma_{\searchset}^1(C_i) / L_i : |L| = |L_i|$. 
The proof of \textbf{l-cover}$_{\searchset}$($L_i$) follows by Lemma~\ref{lemma:gamma-cover}. Showing minimality 
follows from Lemma~\ref{lemma:chaincard} and the observation that $\gamma_{\searchset}^1(C)$, where $|C| = k$, 
creates k-size upper-bound index sets. Therefore,  $|L_i| \leq |C| = k$.  We also know by Lemma~\ref{lemma:set-size} 
that $C' \equiv |\alpha_{\searchset}^{1}(L_i)|$ and $|C'| = |L_i|$. But $C'$ can't be smaller than $C$ as $C$ is in 
$g_{\searchset}$ (min chains) hence $|C| = |L_i|$ and $L_i$ must be minimal and an l-cover.  We therefore conclude it is 
in $f_{\searchset}$
\end{proof}

\end{appendix}

\end{document}